\documentclass[bj,english]{imsart}

\RequirePackage{amsmath,amsthm,amsfonts,amssymb}
\RequirePackage[numbers]{natbib}

\RequirePackage[colorlinks,citecolor=blue,urlcolor=blue]{hyperref}
\RequirePackage{graphicx}

\usepackage{pdfpages}

\startlocaldefs

	\usepackage{lmodern}				

	\usepackage[utf8x]{inputenc}

	\usepackage{babel}					

\RequirePackage{bbm}					
\RequirePackage{stmaryrd}

	\usepackage{pstricks}
	\usepackage{hyperref}
	\usepackage{cleveref}

\newcommand{\ensnombre}[1]{\mathbb{#1}}%
\newcommand{\N}{\ensnombre{N}}

\newcommand{\R}{\ensnombre{R}}

\newcommand{\abs}[1]{\left\lvert #1 \right\rvert}

\newcommand{\defeq}{\mathrel{\mathop:}=}
\newcommand{\eqdef}{\mathrel{=}:}

\newcommand{\ind}{\mathbbm{1}}
\renewcommand{\P}{\mathbb{P}}

\newcommand{\E}{\mathbb{E}}

\def \Var{\hbox{{\textrm{Var}}}}

\def \Cov{\hbox{{\textrm{Cov}}}}

\newcommand\restrict[1]{\raisebox{-.5ex}{$|$}_{#1}}

\theoremstyle{plain}
\newtheorem{thm}{Theorem}[section]
\newtheorem{cor}[thm]{Corollary}
\newtheorem{lem}[thm]{Lemma}
\newtheorem{prop}[thm]{Proposition}

\theoremstyle{remark}
\newtheorem{rem}[thm]{Remark}

\newcommand{\cvloi}{\overset{\mathcal{L}}{\underset{n\to\infty}{\longrightarrow}}}
\newcommand{\egloi}{\overset{\mathcal{L}}{=}}

\newcommand{\pscal}[1]{\langle #1 \rangle}

\hypersetup{%
colorlinks=true,%
linkcolor = black,%
urlcolor = blue,%
citecolor = blue}

\newcommand{\usd}{\frac{1}{2}}

\newcommand{\ysnj}{Y_{\sigma_n(j)}} 
\newcommand{\ysnjpu}{Y_{\sigma_n(j+1)}} 
\newcommand{\sn}{\sigma_n}

\newcommand{\lp}{\left(}
\newcommand{\rp}{\right)}

\newcommand{\Yu}{Y^{\textbf u}}
\newcommand{\Xu}{X^{\textbf u}}

\endlocaldefs

\begin{document}

\begin{frontmatter}
\title{Global Sensitivity Analysis: a novel generation of mighty estimators based on rank statistics}
\runtitle{GSA: mighty estimators and rank statistics}

\begin{aug}
\author[1]{\fnms{Fabrice} \snm{Gamboa}\ead[label=e1]{fabrice.gamboa@math.univ-toulouse.fr}},
\author[2]{\fnms{Pierre} \snm{Gremaud}\ead[label=e2]{gremaud@ncsu.edu}},
\author[3]{\fnms{Thierry} \snm{Klein}\ead[label=e3]{thierry01.klein@enac.fr}},
\and
\author[4]{\fnms{Agnès} \snm{Lagnoux}\ead[label=e4]{lagnoux@univ-tlse2.fr}}

\address[1]{Institut de Math\'ematiques de Toulouse and ANITI; UMR5219. Universit\'e de Toulouse; CNRS. UT3, F-31062 Toulouse, France., \printead{e1}}
\address[2]{Department of Mathematics. NC State University. Raleigh, North Carolina 27695, USA., \printead{e2}}
\address[3]{Institut de Math\'ematiques de Toulouse; UMR5219. Universit\'e de Toulouse; ENAC - Ecole Nationale de l'Aviation Civile , Universit\'e de Toulouse, France, \printead{e3}}
\address[4]{Institut de Math\'ematiques de Toulouse; UMR5219. Universit\'e de Toulouse; CNRS. UT2J, F-31058 Toulouse, France., \printead{e4}}
\end{aug}

\begin{abstract}
~~We propose a new statistical estimation framework for a large family of global sensitivity analysis indices. Our approach is based on rank statistics and uses an empirical correlation coefficient recently introduced by Chatterjee \cite{Chatterjee2019}. We show how to apply this approach  to compute not only the Cram\'er-von-Mises indices, directly related to Chatterjee's notion of correlation, but also first-order Sobol' indices, general metric space indices and higher-order moment indices. We establish consistency of the resulting estimators and demonstrate their numerical efficiency, especially for small sample sizes. In addition, we prove a central limit theorem for the estimators of the first-order Sobol' indices.
\end{abstract}


\end{frontmatter}


\section{Introduction}\label{sec:into}

The use of complex computer models for the analysis of applications from the sciences, engineering and other fields is by now routine. Often, the models are expensive to run in terms of computational time. It is thus crucial to understand, with just a few runs, the global influence of  one or several inputs  on  the  output of the system under study  \cite{santner2003design}. When these inputs are regarded as random elements, this   problem  is  generally  referred to as  Global Sensitivity Analysis (GSA).
  We refer to \cite{rocquigny2008uncertainty,saltelli-sensitivity,sobol1993}
 for an overview of the practical aspects of GSA.

A  popular and highly useful tool to quantify input influence is the  Sobol' indices. These indices were first introduced in \cite{sobol2001global} and are well tailored to the case of scalar outputs (and even to the case of vectorial and functional outputs). Thanks to the Hoeffding decomposition  \cite{Hoeffding48}, the Sobol' indices compare the conditional variance of the output knowing some of the input variables to the total variance of the output. 
Since Sobol' indices are variance based, they only quantify the second-order influence of the inputs. Many authors proposed other criteria to compare the conditional distribution of the output knowing some of the inputs to the distribution of the output (see, e.g., higher moments indices in \cite{Owen13,ODC13,Owen12}, indices using divergences or distances between measures in \cite{borgonovo2007new,borgonovo2011moment, DaVeiga13}, goal-oriented indices using contrast functions in \cite{FKR13}, distribution-based indices as Cramér-von-Mises indices in \cite{GKL18}). 


Many different estimation procedures of the Sobol' indices have been proposed and studied. Some estimation procedures are based on different designs of experiment using for example polynomial chaos  (see \cite{Sudret2008global} and the reference therein  for more details). Some other natural procedures are based on Monte-Carlo or quasi Monte-Carlo design of experiments (see \cite{Kucherenko2017different,Owen13} and references therein for more details). 
In particular, an efficient estimation of the Sobol' indices  can be performed through the so-called Pick-Freeze method. See Section \ref{ssec:def:Sobol'} below for  its description. Observe that the Pick-Freeze estimation procedure allows the estimation of  several sensitivity indices: the classical Sobol' indices for real-valued outputs, as well as their generalization for vectorial-valued codes, but also the indices based on higher moments  \cite{ODC13} and 
the Cram\'er-von-Mises indices which take into account on the whole distribution (see \cite{GKL18,FGM2017} and Section \ref{ssec:def:cvm} below for more details on such indices). In addition, the Pick-Freeze estimators have desirable statistical properties such as consistency, central limit theorem (CLT) with a rate of convergence in $\sqrt n$, concentration inequalities and Berry-Esseen bounds, and asymptotic efficiency (see \cite{janon2012asymptotic,pickfreeze} and Section \ref{ssec:def:Sobol'} below for more details). However, the Pick-Freeze scheme  has  two major  drawbacks. First, it relies on a particular experimental design that may be unavailable in practice. Second, its cost may be prohibitive when estimating several indices. Naturally, the cost of an estimator depends on the cost of each evaluation of the code and on the number of evaluations. The number of model calls to estimate all first-order Sobol' indices
grows linearly with the number of input parameters.  For example, if we consider $p=99$ input parameters  and only  $n=1000$ calls are allowed, then only a sample of size $n/(p+1)=10$ is available to estimate each single first-order Sobol' index. It is a poor amount of information to get a satisfying estimation of the Sobol' indices. 

In a recent work \cite{Chatterjee2019}, Chatterjee  studies the dependence between two variables by introducing an empirical correlation coefficient based on rank statistics, see Section \ref{ssec:estim_cvm_chatterjee} below for the precise definition. Further, the quantification of the dependence has also been investigated in the bivariate case (namely, in the copula setting), see \cite{trutschnig2011strong,dette2013copula,azadkia2019simple}. The striking point of \cite{Chatterjee2019} is that this empirical correlation coefficient converges almost surely (a.s.) to the Cram\'er-von-Mises index priorly introduced in \cite{GKL18} as the sample size goes to infinity. 

In this paper, we show how to embed Chatterjee's method in the GSA framework, thereby eliminating the two drawbacks of the classical Pick-Freeze estimation mentioned above. Thus no particular design of experiment is needed for the estimation that can be done with a unique $n$-sample. In addition, we generalize Chatterjee's approach to allow the estimation of a large class of GSA indices which includes the Sobol' indices and the higher-order moment indices proposed by Owen \cite{Owen13,ODC13,Owen12} (see Section \ref{ssec:def:Sobol'} below). 
Using a single sample of size $n$, it is now possible to estimate at the same time all the first-order Sobol' indices, the Cram\'er-von-Mises indices, and other useful sensitivity indices. 
Furthermore, we show  that this new procedure provides estimators also converging at rate $\sqrt n$ by proving a CLT in the estimation of the first-order Sobol' indices. 

The paper is organized as follows. In Section \ref{sec:Sobol'_cvm}, we recall the context of GSA, the definition of the Sobol' indices and Cram\'er-von-Mises indices, and their classical Pick-Freeze estimations. Section \ref{sec:chat} focuses on Chatterjee's method, called rank-based method in this paper. More precisely, we show how the Cramér-von-Mises indices can be also estimated using the rank-based method (Section \ref{ssec:estim_cvm_chatterjee}) and we present its generalization to estimate sensitivity indices together with the consistency of the estimation procedure (Section \ref{sec:gene}). Section \ref{ssec:Sobol'} is dedicated to Sobol' indices.  We prove the asymptotic normality of their estimators based on rank statistics. In addition, we propose a comparison of the different estimation procedures in Section \ref{ssec:comparison} while 
Section \ref{sec:known_indices} considers other classical sensitivity indices.  
Section \ref{sec:exnum} is dedicated to a numerical comparison between the Pick-Freeze estimation procedure and the rank-based method. We first compare the numerical performances of both estimators on a linear model. Finally, we consider a real life application. As expected, the rank-based estimation method outperforms the classical Pick-Freeze procedure, even for small sample sizes (which are  common  in practice). 
Conclusions and perspectives are offered in Section \ref{sec:conclu}.

After a first submission of this paper, we have been aware of the very nice work of Broto {\it et al} \cite{broto2020variance} concerning the statistical estimation of Shapley effect where the use of closest neighbors is also put in action to built consistent estimates. We also notice that there is actually a strong scientific interest around asymptotic behavior for the statistical method introduced in   
\cite{Chatterjee2019}. Indeed, during the revision of this paper, we have a look on the very nice paper \cite{auddy2021exact} where an asymptotic contiguity study is performed. 

\section{Global sensitivity analysis and Pick-Freeze estimation}
\label{sec:Sobol'_cvm}

\subsection{Sobol' indices}
\label{ssec:def:Sobol'}

\paragraph*{Context and definition of the Sobol' indices}

The quantity of interest (QoI) $Y$ is obtained from  the numerical code and is regarded as a function $f$ of the vector of the distributed input $(X_i)_{i=1,\ldots,p}$
\begin{equation}\label{def:model}
Y=f(X_{1},\ldots, X_{p}),
\end{equation}
where $f$ is defined on the state space  $E_1\times \ldots \times E_p$, $X_i \in E_i$, $i=1,\dots,p$. Classically, the $X_i$'s  are assumed to be independent random variables and a sensitivity analysis is performed using the Hoeffding decomposition \cite{anton84,van2000asymptotic} leading to the standard Sobol' indices \cite{sobol1993}. This assumption is made throughout the paper, unless explicitly stated otherwise. More precisely, assume  $f$ to be real-valued and square integrable and let $\textbf u$ be a subset of  $\{1, \ldots, p\}$ and $\sim\!
\textbf u$ its complementary set in  $\{1, \ldots, p\}$. 
Setting $X_{\textbf u}=(X_i, i \in\textbf u)$ and $X_{\sim \textbf u}=(X_i, i \in\sim\!\textbf u)$,
 the corresponding Sobol' indices take the form 
\begin{equation}\label{eq:defsob}
S^{\textbf u}=\frac{\Var\left(\mathbb{E}[Y|X_{\textbf u}]\right)}{\Var(Y)} \quad \text{and} \quad S^{\sim \textbf u}=\frac{\Var\left(\mathbb{E}[Y|X_{\sim \textbf u}]\right)}{\Var(Y)}.
\end{equation}

By definition, the Sobol' indices  quantify the fluctuations of the output $Y$ around its mean. When the practitioner is not interested in the mean behavior of $Y$ but rather in its median, in its tail, or even in its quantiles, the Sobol' indices become less appropriate to quantify sensitivity. GSA must then be performed in a framework which takes into account more than one specific moment, such as the variance for Sobol' indices.

\paragraph*{Pick-Freeze estimation procedure of the Sobol' indices}

 A Monte-Carlo scheme can be used to estimate the Sobol' indices. The corresponding   Pick-Freeze approach from \cite{pickfreeze,GKL18,janon2012asymptotic} relies on expressing the variances of the conditional expectations in terms of covariances which are easily and well estimated by their empirical versions. 
To that end, we define, for any subset $\textbf u$ of $\{1,\dots,p\}$
\begin{align}\label{def:pf}
Y^{\textbf u}\defeq f(\Xu).
\end{align}
where $\Xu$ is such that $X_{\mathbf u}^{\mathbf u} = X_{\mathbf u}$
and  $\Xu_i=X'_i$ if $i\in \sim \textbf u$, $X'_i$ being an independent copy of $X_i$. The estimation procedure relies on the following result
\begin{align}\label{eq:pf}
&\Var(\E[Y|X_{\textbf u}])=\Cov(Y,\Yu).
\end{align}
The reader is referred to \cite[Lemma 1.2]{janon2012asymptotic} for its proof.

The natural estimator of $S^{\mathbf u}$ is then given by
\begin{align}\label{def:Sn}
S_n^{\mathbf u}=\frac{
\frac{1}{n}\sum_{j=1}^nY_jY_j^{\mathbf u}-\left(\frac{1}{n}\sum_{j=1}^nY_j\right)\left(\frac{1}{n}\sum_{j=1}^nY_j^{\mathbf u}\right)}
{
\frac{1}{n}\sum_{j=1}^n(Y_j)^2-\left(\frac{1}{n}\sum_{j=1}^nY_j\right)^2
}.
\end{align}
A slightly different estimator that uses all the information available is introduced in \cite{janon2012asymptotic}:
\begin{align}\label{def:Tn}
T_n^{\mathbf u}=\frac{
\frac{1}{n}\sum_{j=1}^nY_jY_j^{\mathbf u}-\left(\frac{1}{n}\sum_{j=1}^n\frac{Y_j+Y_j^{\mathbf u}}{2}\right)^2}
{
\frac{1}{n}\sum_{j=1}^n\frac{(Y_j)^2+(Y_j^{\mathbf u})^2}{2}-\left(\frac{1}{n}\sum_{j=1}^n\frac{Y_j+Y_j^{\mathbf u}}{2}\right)^2
}.
\end{align}

\paragraph*{Asymptotic study} Such 
estimation procedures have been proved to be consistent and asymptotically normal (i.e.\ the rate of convergence is $\sqrt{n}$) in \cite{janon2012asymptotic,pickfreeze}. The limiting variances can be computed explicitly, allowing the practitioner to build confidence intervals. 
In addition, the sequence of estimators $(T_{n}^{\mathbf u})_n$ is asymptotically efficient to estimate $S^{\mathbf u}$ from such a design of experiment (see, \cite{van2000asymptotic} for the definition of the asymptotic efficiency and \cite{pickfreeze} for the details of the result).

\subsection{Cram\'er-von-Mises indices}
\label{ssec:def:cvm}

\paragraph*{Definition of the Cram\'er-von-Mises indices} The Cram\'er-von-Mises indices introduced in \cite{GKL18} provide alternative indices based on the whole distribution rather than on the second moment of the output $Y$ only. The main idea of Cramér-von-Mises indices is to compare the conditional cumulative distribution function (c.d.f.) to the unconditional one via the $L^2$-norm. As for the Sobol' indices, they compare the conditional expectation of the output to the unconditional one. Notably, they are constructed following a similar scheme so that any procedure that 
estimates one index can be adapted to estimate the other.

More precisely, the Cramér-von-Mises indices are defined by 
\begin{align}\label{eq:defCVM}
S_{2,CVM}^{\textbf u}= \frac{\int_{\R}\E\left[\left(F(t)-F^{\textbf u}(t)\right)^{2}\right]dF(t)}{\int_{\R} F(t)(1-F(t))dF(t)} 
\end{align}
where $F$ is the cumulative distribution function of $Y$
\[
F(t)=\P\left(Y\leqslant t\right)=\E\left[\ind_{\{Y\leqslant t\}}\right] \quad  (t\in \R)
\]
and $F^{\textbf u}$ is its Pick-Freeze version:
\[
F^{\textbf u}(t)=\P\left(Y\leqslant t|X_{\textbf u}\right)=\E\bigl[\ind_{\{Y\leqslant t\}}|X_{\textbf u}\bigr] \quad  (t\in \R).
\]
This definition stems from the Hoeffding decomposition of the collection of r.v.\ $(\ind_{\{Y\leqslant t\}})_{t\in \R}$.

\paragraph*{Pick-Freeze estimation procedure of the Cram\'er-von-Mises indices}
The estimation procedure relies on \eqref{eq:pf} with $Y\leftarrow \ind_{\{Y\leqslant t\}}$:
\begin{align}\label{eq:pf_cvm}
&\Var(\E[\ind_{\{Y\leqslant t\}}|X_{\textbf u}])=\Cov(\ind_{\{Y\leqslant t\}},\ind_{\{\Yu\leqslant t\}}).
\end{align}
Consequently, the Monte-Carlo estimation can be done as follows.
 In addition to the classical design of experiment required to estimate the Sobol' indices (an $n$-sample $(Y_1,\ldots,Y_n)$ of the output $Y$ and an $n$-sample $(\Yu_1,\ldots,\Yu_n)$ of its Pick-Freeze version   $\Yu$), a third independent $n$ sample $(W_1,\ldots,W_n)$ of the output $Y$ is necessary in order to deal with the integral with respect to $dF(t)$ in \eqref{eq:defCVM}. Then the empirical estimator of $S_{2,CVM}^{\mathbf u}$ is
\begin{align}\label{def:cvm_n}
\frac{\frac{1}{n}\sum_{k=1}^n\left(
\frac{1}{n}\sum_{j=1}^n\ind_{\{Y_j\leqslant W_k\}}\ind_{\{\Yu_j \leqslant W_k\}}-\frac{1}{n}\sum_{j=1}^n\ind_{\{Y_j\leqslant W_k\}}\frac{1}{n}\sum_{j=1}^n\ind_{\{\Yu_j\leqslant W_k\}}\right)}
{\frac{1}{n}\sum_{k=1}^n\left(
\frac{1}{n}\sum_{j=1}^n\ind_{\{Y_j\leqslant W_k\}}-\left(\frac{1}{n}\sum_{j=1}^n\ind_{\{Y_j\leqslant W_k\}}\right)^2\right)
}.
\end{align}

\paragraph*{Asymptotic study}  As showed in \cite{GKL18}, this estimator is consistent and asymptotically Gaussian (i.e.\ the rate of convergence is $\sqrt{n}$). The limiting variance can be computed explicitly, allowing the practitioner to build confidence intervals.

\section{A novel generation of estimators based on rank statistics}\label{sec:chat}

\subsection{Chatterjee's correlation coefficient} \label{ssec:estim_cvm_chatterjee}

In \cite{Chatterjee2019}, Chatterjee considers a pair of real-valued random variables $(V,Y)$ and an i.i.d.\ sample $(V_j,Y_j)_{1\leqslant j\leqslant n}$. In order to simplify the presentation, we assume that the  laws of $V$ and $Y$ are both diffuse (ties are excluded). The pairs $(V_{(1)},Y_{(1)}),\ldots,(V_{(n)},Y_{(n)})$ are rearranged
in such a way that 
\[
V_{(1)}< \ldots< V_{(n)}.
\]
Then let $\pi(j)$ be the rank of $V_j$ in the sample $(V_1,\dots, V_n)$ of $V$ and define 
\begin{align}\label{def:N}
N'(j)=\begin{cases}
\pi^{-1} (\pi(j)+1) & \text{if  $\pi(j)+1\leqslant n$},\\
j & \text{if  $\pi(j)=n$}.\\
\end{cases}
\end{align}

The new correlation coefficient defined by Chatterjee in \cite{Chatterjee2019} is denoted $\xi_n(V,Y)$ and given by
\begin{equation}\label{eq:coeffChatterjee}
\frac{1}{n}\sum_{j=1}^n  \Bigl(\frac{1}{n}\sum_{k=1}^n \ind_{\{Y_k\leqslant Y_j\}}\ind_{\{Y_k\leqslant Y_{N'(j)}\}}- \Bigl(\frac{1}{n}\sum_{k=1}^n \ind_{\{Y_j\leqslant Y_k\}}\Bigr)^2\Bigr)\Big/\frac{1}{n}\sum_{j=1}^n F_n(Y_j)(1-F_n(Y_j))
\end{equation} 
where $F_n$ stands for the empirical distribution  function of $Y$: $F_n(t)=\frac 1n \sum_{k=1}^n \ind_{\{Y_k\leqslant t\}}$.

The author proves that $\xi_n(V,Y)$ converges a.s.\ to a deterministic limit $\xi(V,Y)$ which is equal to the Cramér-von-Mises sensitivity index $S_{2,CVM}^V$ with respect to $V$ as soon as $V$ is one of the random variables $X_1$, ..., $X_p$ in the model \eqref{def:model} that are assumed to be real-valued. 
Further, he also proves a CLT when $V$ and $Y$ are independent.

Observe that the analogue of the Pick-Freeze version $Y^V$ with respect to $V$  of $Y$ becomes $Y_{N}$ and \eqref{eq:pf_cvm} is replaced by the formula
\begin{equation}\label{eq:trick_chatterjee}
\E[\ind_{\{Y_j\geqslant t\}}\ind_{\{Y_{N'(j)}\geqslant t\}}|V_1,\ldots, V_n]=G_{V_j}(t)G_{V_{N'(j)}}(t)
\end{equation}
for all $j=1,\dots,n$ that is mentioned in the proof of Lemma 7.10 in \cite[p.24]{Chatterjee2019}, with $G_{V}$ the conditional survival function: $G_{V}(t)=\P(Y\geqslant t| V)$.

It is worth noticing that a unique $n$ sample of input-output provides consistent estimations of the $p$ first-order Cram\'er-von-Mises indices.

\subsection{Generalization of Chatterjee's method}\label{sec:gene}


In this section, we propose a universal estimation procedure of expectations of the form 
\[
\E[\E[g(Y)|V]\E[h(Y)|V]],
\]
for two integrable functions $g$ and $h$. In fact, we consider a more general random element $V$ (no longer assumed to be real) and a more general permutation denoted by $\tau_n$.  
This result is a generalization of \eqref{eq:trick_chatterjee} and can be interpreted as an approximation of \eqref{eq:pf}. To this end, we introduce the function $\Psi_V$ defined by
\begin{align}\label{def:Psi}
\Psi_V(g)=\E[g(Y)|V]
\end{align}
for any integrable function $g$. Let $\mathcal F_n$ be the $\sigma$-algebra generated by $\{V_1,\dots,V_n\}$. Note that in Section \ref{ssec:estim_cvm_chatterjee}, we have considered $g(x)=g_t(x)=\ind_{\{x\geqslant t\}}$ so that $\Psi_V(g)=\P(Y\geqslant t| V)=G_V(t)$.

\begin{lem}\label{lem:pfChatt} 
Let $g$ and $h$ be two integrable functions such that $gh$ is also integrable. Let $(V_j,Y_j)_{1\leqslant j\leqslant n}$ be an $n$-sample of $(V,Y)$. Consider a $\mathcal F_n$-measurable random permutation $\tau_n$ such that $\tau_n(j)\neq j$, for all $j=1,\dots,n$. Then
\begin{equation}\label{eq:pfChatt}
\E\left[g(Y_j)h(Y_{\tau_n(j)})|V_1,\ldots, V_n\right]=\Psi_{V_j}(g)\Psi_{V_{\tau_n(j)}}(h).
\end{equation}
\end{lem}

The previous lemma (the proof of which has been postponed to Appendix  \ref{app:cons}) leads to a generalization of the first part of the numerator of $\xi_n$ defined in 
\eqref{eq:coeffChatterjee}. Following the same lines as in \cite{Chatterjee2019}, one may prove that such a quantity converges a.s. as $n\to \infty$ under some mild conditions. The reader is referred to Appendix \ref{app:cons} for the detailed proof of Proposition \ref{prop:cv}.

\begin{prop}\label{prop:cv}
Let $g$ and $h$ be two bounded measurable functions. Consider a $\mathcal F_n$-measurable random permutation $\tau_n$ with no fix point (i.e.\ $\tau_n(j)\neq j$ for all $j=1,\dots,n$) and such that $V_{\tau_n(i)}\overset{\mathcal L}{=} V_{\tau_n(j)}$ for any $i$ and $j=1,\dots,n$. In addition, we assume that for any  $j=1,\dots,n$, $V_{\tau_n(j)}\to V_j$ as $n\to \infty$ a.s.
Then 
$\chi_n(V,Y;g,h)$ defined by
\begin{align}
\chi_n(V,Y;g,h) =& \frac{1}{n}\sum_{j=1}^n g(Y_j)h(Y_{\tau_n(j)})\label{eq:coeffChatterjee_version_general_num}
\end{align}
 converges a.s. as $n\to \infty$
to $
\chi(V,Y;g,h)= \E[\Psi_{V}(g)\Psi_V(h)]$,
where $\Psi_{V}$ has been defined in \eqref{def:Psi}.
\end{prop}

Notice that the permutation $\tau_n=N$ defined by \begin{align}\label{def:Nnous}
N(j)=\begin{cases}
\pi^{-1} (\pi(j)+1) & \text{if  $\pi(j)+1\leqslant n$},\\
\pi^{-1} (1) & \text{if  $\pi(j)=n$}.\\
\end{cases}
\end{align}
satisfies the assumptions of Lemma \ref{lem:pfChatt} and  Proposition \ref{prop:cv}. Observe that $N$ only differs from $N'$ defined in \eqref{def:N}  at $j$ such that $\pi(j)=n$.

\section{The rank estimator of the first-order Sobol' indices}\label{ssec:Sobol'}

\subsection{Estimation procedure based on rank statistics}

We can now leverage the above results and construct a new family of estimators for Sobol' indices. 
More precisely, let us consider the model \eqref{def:model} and assume we want to estimate the first-order Sobol' index $S^1$ defined in \eqref{eq:defsob} with respect to $V=X_1$ assumed to be real-valued. 
We then define $N$ as in \eqref{def:Nnous} where $\pi$ is the rank of $X_1$.  
Taking $g(x)=h(x)=x$ and $\tau_n=N$,  
\eqref{eq:pfChatt} provides the analogue to $\xi_n$ to estimate the classical Sobol' indices:
\begin{align}\label{eq:coeffChatterjee_version_Sobol'}
\xi_n^{\text{Sobol'}}(X_1,Y) \defeq \frac{\frac 1n\sum_{j=1}^n Y_jY_{N(j)}-\left(\frac 1n \sum_{j=1}^n Y_j\right)^2}{\frac 1n \sum_{j=1}^n (Y_j)^2-\left(\frac 1n \sum_{j=1}^n Y_j\right)^2},
\end{align}
where the denominator is reduced to the empirical variance of $Y$. 
As the functions $g$ and $h$ are here unbounded, Proposition \ref{prop:cv} does not apply and thus offers no asymptotic information. However, the quantity of interest $Y$ being generally bounded in practice, appropriately truncated versions of $g$ and $h$ could be considered.

\subsection{A central limit theorem}\label{ssec:CLT}

We establish a CLT for the estimator $\xi_n^{\text{Sobol'}}(X_1,Y)$ of the first-order Sobol' index 
with respect to $X_1$ (assumed to be real-valued) under some mild assumptions on the model $f$ and the random input $X_1$ in \eqref{def:model}. The proof of the theorem is given in Appendix \ref{app:asymp_norm}. 

\begin{thm}\label{th:tcl}
Assume that $X_1$ is uniformly distributed on $[0,1]$ and $f$ in \eqref{def:model} is a twice differentiable function with respect to its first coordinate. Further, we suppose that $f$ and its two first derivatives (with respect to its first coordinate) are  bounded. Then 
\[
\sqrt n \left(\xi_n^{\text{Sobol'}}(X_1,Y)-S^1\right)
\]
is asymptotically Gaussian with zero mean and explicit variance $\sigma^2$ given in Appendix \ref{app:var_tcl}.
\end{thm}

\begin{rem}
The boundedness of $f$ implies that $f$ has a fourth moment, that is the minimal assumption to get a CLT. 

Moreover, let us observe that Theorem \ref{th:tcl} only implies the convergence in probability. Nevertheless, under the assumptions of Theorem \ref{th:tcl} ($f$ bounded so is $Y$), Proposition \ref{prop:cv} applies to derive the almost sure convergence
of $\xi^{Sobol}_n (X1, Y )$. 
\end{rem}

The assumption on the distribution of $X_1$ can be relaxed as stated in the following corollary.

\begin{cor}\label{cor:tcl}
Let $F_{X_1}$ be the cumulative distribution function of $X_1$.  Assume that $f\circ F_{X_1}^{-1}$ is a twice differentiable function such that $f\circ F_{X_1}^{-1}$ and its two first derivatives are bounded. Then the conclusion of Theorem \ref{th:tcl} still holds. 
\end{cor}

Theorem \ref{th:tcl} and Corollary \ref{cor:tcl} naturally allow to build statistical tests for testing 
$
H_0: S^1=0 \quad \text{against} \quad  H_1: S^1\not=0.
$
One can note that Chatterjee \cite{Chatterjee2019} result  allows to test the independence of the input $X_1$ with respect to the output $Y$ which is a stronger assumption than $S^1=0$, this was for example studied in \cite{shi2020power}. In addition, our result allows to compute the power of the statistical test against any alternative of the kind  $H_{1,0}: S^1>s^1_0$ for any $s^1_0>0$.

\begin{rem}
A careful reading of the different steps of the proof shows that Theorem \ref{th:tcl} can be slightly extended to more general situations involving more than two successive order statistics and with more general second variable $(X_2,\ldots,X_p)$.
See the forthcoming paper \cite{GKL21}.
\end{rem}

The proof of our CLT is a bit long and technical and is postponed to the Appendix \ref{app:asymp_norm}. In a nutshell, this proof stands on three main ingredients. First, the regularity assumption on the function $f$ allows to expand the statistic under study as a quadratic functional of the two independent sequences of random variables. The quadratic part for the first sequence involves order statistics of the uniform distribution and may be linearized. The second ingredient is the distribution representation of uniform order statistics by ratios of exponential convolution. The third ingredient is less classical and involves a conditional trick to show a central limit theorem for an empirical mean of a product. Let sketch the idea on a simple example. Let $(\xi_n)_n$ and $(\delta_n)_n$ be two independent sequences of centered square integrable random variables. We set $M_n=n^{-1/2}\sum_{j=1}^{n}\xi_j\delta_j$ and let $\mathcal{T}$ be the $\sigma$-field generated by the sequence $(\delta_n)$. Of course, the classical CLT gives that $M_n$ converges in distribution towards a centered Gaussian distribution with variance $\Var(\xi_1)\Var(\delta_1)$. A less classical proof of this result consists in showing that, a.s., conditionally to $\mathcal{T}$ the same convergence in distribution holds. Indeed, this last result follows directly from the Lindeberg CLT and the strong law of large numbers for 
$n^{-1}\sum_{j=1}^n\delta_j^2$.  
    
\subsection{Comparison of the different estimation procedures}\label{ssec:comparison}

The estimator based on rank statistics $\xi_n^{\text{Sobol'}}(X_1,Y)$ defined in \eqref{eq:coeffChatterjee_version_Sobol'} can be compared to the classical Pick-Freeze estimators $S_n^1$  and $T_n^1$ given in \eqref{def:Sn} and \eqref{def:Tn} respectively (with $\mathbf u=\{1\}$) but also to a sequence of estimators involving the estimators $\widehat T_n$ introduced in \cite{da2008efficient}. 

\paragraph*{Required sample sizes}
With the rank-based procedure, a unique $n$-sample of input-output provides consistent and asymptotically normal estimations of the $p$ first-order Sobol' indices (together with consistent and asymptotically normal estimations of the $p$ first-order Cram\'er-von-Mises indices with no extra cost). In contrast, using the Pick-Freeze estimation, if one wants to estimate all the $p$ first-order Sobol' indices and the $p$ Cram\'er-von-Mises indices, $(p+2)n$ calls of the computer code are required.
The number of calls  grows linearly with respect to the number of input parameters. This is a practical issue for large input dimension domains.
A second drawback of the Pick-Freeze estimation scheme comes from the need of the particular Pick-Freeze design that is not always available.

\paragraph*{Limiting variances} Since the empirical mean and  variance are already known to be asymptotically efficient in the statistical sense\footnote{The reader is referred to \cite[Section 25]{van2000asymptotic} for the definition of the asymptotic efficiency and related results.} to estimate the expectation and the variance of the output, we restrict our study to the comparison of the limiting variances obtained via the Pick-Freeze and the rank-based procedures in the estimation of $\E[\E[Y|X_1]^2]$ only.

In view of the proof of \cite[Proposition 2.2]{janon2012asymptotic}, the Pick-Freeze limiting variance obtained using both $S_n^1$ and $T_n^1$ in estimating $\E[\E[Y|X_1]^2]=\E[YY^1]$ is simply given by
$\Var(YY^1)$,
where $Y^1=f(X_1,W^1)$ is the Pick-Freeze version of $Y=f(X_1,X_2,\ldots,X_p)=f(X_1,W)$.

Using the above Lemmas \ref{lem:B} and \ref{lem:C} together with \eqref{eq:CLT_somme} leads to the rank-based limiting variance obtained using $\xi_n^{\text{Sobol'}}(X_1,Y)$:
\begin{align}\label{eq:lim_var_rank}
\Sigma_B^{1,1}+\Sigma_C^{1,1}&=\E\left[\Var\left(YY^1|X_1\right) \right]
+ \E\left[\Cov\left(YY^1,YY^{11}|X_1\right)\right]-\E[(Y+Y^1)f_x(X_1,W) X_1]^2\nonumber\\
&\qquad +\E[(Y+Y^1)(\tilde Y+\tilde Y^1)f_x(X_1,W) f_x(\tilde X_1,\tilde W) (X_1\wedge \tilde X_1)],
\end{align}
where $Y=f(X_1,X_2,\ldots,X_p)=f(X_1,W)$, $Y^1=f(X_1,W^1)$, $Y^{11}=f(X_1,W^{11})$, $\tilde Y=f(\tilde X_1,\tilde W)$, and $\tilde Y^1=f(\tilde X_1,\tilde W^1)$ with $X_1$ and $\tilde X_1$ i.i.d., $W$, $\tilde W$, $W^1$, and $W^{11}$ i.i.d.\ also independent of $X_1$ and $\tilde X_1$. Note that $Y^1$ and $Y^{11}$ (respectively $\tilde Y^1$) are Pick-Freeze versions of $Y$ (resp. $\tilde Y$). The paragraph's aim is to compare the limiting variances obtained by the two methods (Pick-Freeze and rank-based).

To do so, we recall that the Pick-Freeze experiment requires $n(p+1)$ observations (or computations of the black-box code) to estimate the $p$ first-order Sobol' indices. In order to have a fair comparison of both estimation methods, we then consider that we  have $n(p+1)$ i.i.d.\ observations of $Y$ given by model \eqref{def:model} to estimate the $p$ first-order Sobol' indices using the rank statistics. With $n(p+1)$ observations instead of $n$, the asymptotic variance obtained using the rank-based methodology is divided by $(p+1)$, so that we want to 
compare
\[
V_{\text{PF}}\defeq (p+1)(\Var(YY^1) ,\ldots,\Var(YY^p) )^\top
\text{ 
to }
V_{\text{Rank}}\defeq (\Sigma_B^{1,1}+\Sigma_C^{1,1},\ldots, \Sigma_B^{p,p}+\Sigma_C^{p,p})^\top
\]
where $Y^i$ is the Pick-Freeze version of $Y$ with respect to $X_i$ (for $i=2,\ldots,p$) and $\Sigma_B^{i,i}+\Sigma_C^{i,i}$ has the same expression as $\Sigma_B^{1,1}+\Sigma_C^{1,1}$ in \eqref{eq:lim_var_rank} replacing the superscripts and the subscripts $1$ by $i$ (for $i=2,\ldots,p$).


\textbf{Example}. 
 We consider the following linear model
\begin{align}\label{def:linear_model}
Y=f(X_1,\ldots,X_p)=\alpha X_1 + X_2 + \ldots + X_p,
\end{align}
where $\alpha>0$ is a fixed constant, $X_1$, $X_2$, $\ldots$, and $X_p$ are $p$ independent and uniformly distributed random variables on $[0,1]$.

We denote by $m_{1,p}$ and $m_{2,p}$ the two first moments of $Z_p\defeq X_2+\ldots+X_p$ and $m_{1,p,\alpha}$ and $m_{2,p,\alpha}$ the two first moments of $Z_{p,\alpha}\defeq \alpha X_1+X_3+\ldots+X_p$. In addition, let $v_p$ and $v_{p,\alpha}$ be the variances of $Z_p$ of $Z_{p,\alpha}$. Hence $v_p=m_{2,p}-m_{1,p}^2$, $v_{p,\alpha}=m_{2,p,\alpha}-m_{1,p,\alpha}^2$,
\begin{align*}
&m_{1,p}=\frac{1}{2}(p-1), \quad 
m_{2,p}=\frac{1}{12} (p-1) (3p-2), \quad m_{1,p,\alpha}=\frac 12 (\alpha+m_{1,p-1})=\frac 12 (\alpha+p-2),\\
&m_{2,p,\alpha}=\frac 13 \alpha^2+ \alpha m_{1,p-1}+m_{2,p-1}=\frac 13 \alpha^2+\frac 12 (p-2)\alpha+\frac{1}{12}(p-2)(3p-5).
\end{align*}
By symmetry, after obvious computations,
one gets, for $i=2,\ldots,p$,
\begin{align*}
\Var(YY^1)&= \frac{4}{45}\alpha^4+\frac{1}{3}m_{1,p}\alpha^3+\frac{1}{3}\Bigl(2v_{p}+m_{1,p}^2\Bigr)\alpha^2+2m_{1,p}v_{p}\alpha+v_p (v_p+2m_{1,p}^2),\\
\Var(YY^i)&= \frac{4}{45}+\frac{1}{3}m_{1,p,\alpha}+\frac{1}{3}\Bigl(2v_{p,\alpha}+m_{1,p,\alpha}^2\Bigr)+2m_{1,p,\alpha}v_{p,\alpha}+v_{p,\alpha}(v_{p,\alpha}+2m_{1,p,\alpha}^2)
\end{align*}
while 
\begin{align*}
V_{\text{Rank}}^{1} & =  \frac{4}{45}\alpha^4+\frac{1}{3}m_{1,p}\alpha^3+\frac{1}{3}\Bigl(4v_p+m_{1,p}^2\Bigr)\alpha^2+4m_{1,p}v_p\alpha+v_p\Bigl(v_p+4m_{1,p}^2\Bigr),\\
V_{\text{Rank}}^{i} & = \frac{4}{45}+\frac{1}{3}m_{1,p,\alpha}+\frac{1}{3}\Bigl(4v_{p,\alpha}+m_{1,p,\alpha}^2\Bigr)+4m_{1,p,\alpha}v_{p,\alpha}+v_{p,\alpha}\Bigl(v_{p,\alpha}+4m_{1,p,\alpha}^2\Bigr).
\end{align*}

We compare these limiting variances in Figures \ref{fig:toy_comp} and \ref{fig:toy_diff}. The results are clear and illustrate the fact that the rank-based methodology works much better for all value of $p\geqslant 2$. In addition, the more the value of $p$ increases the greater the gain, as expected.

\begin{figure}[h!]
\includegraphics[scale=0.48]{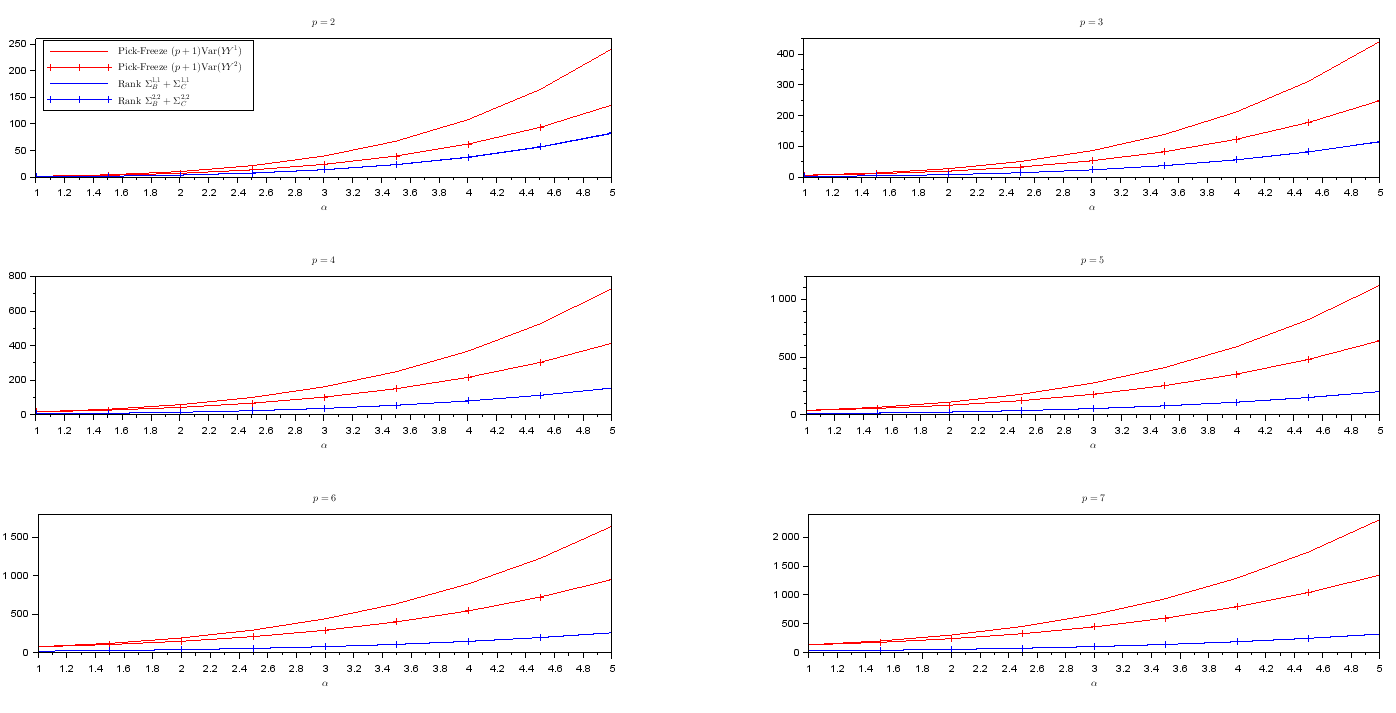}
\caption{Linear model defined in \eqref{def:linear_model}. The limiting variances with respect to $X_1$ (plain lines) and to $X_2$ (plain lines with +) are plotted. The rank-based estimation procedure is represented in blue while the Pick-Freeze estimation procedure is represented in red. As explained, the Pick-Freeze estimation procedure has been weighted by $(p+1)$ to have a fair comparison. The number of variables involved in the model varies from $p=2$ to $p=7$.}
\label{fig:toy_comp}
\end{figure}

\begin{figure}[h!]
\includegraphics[scale=0.26]{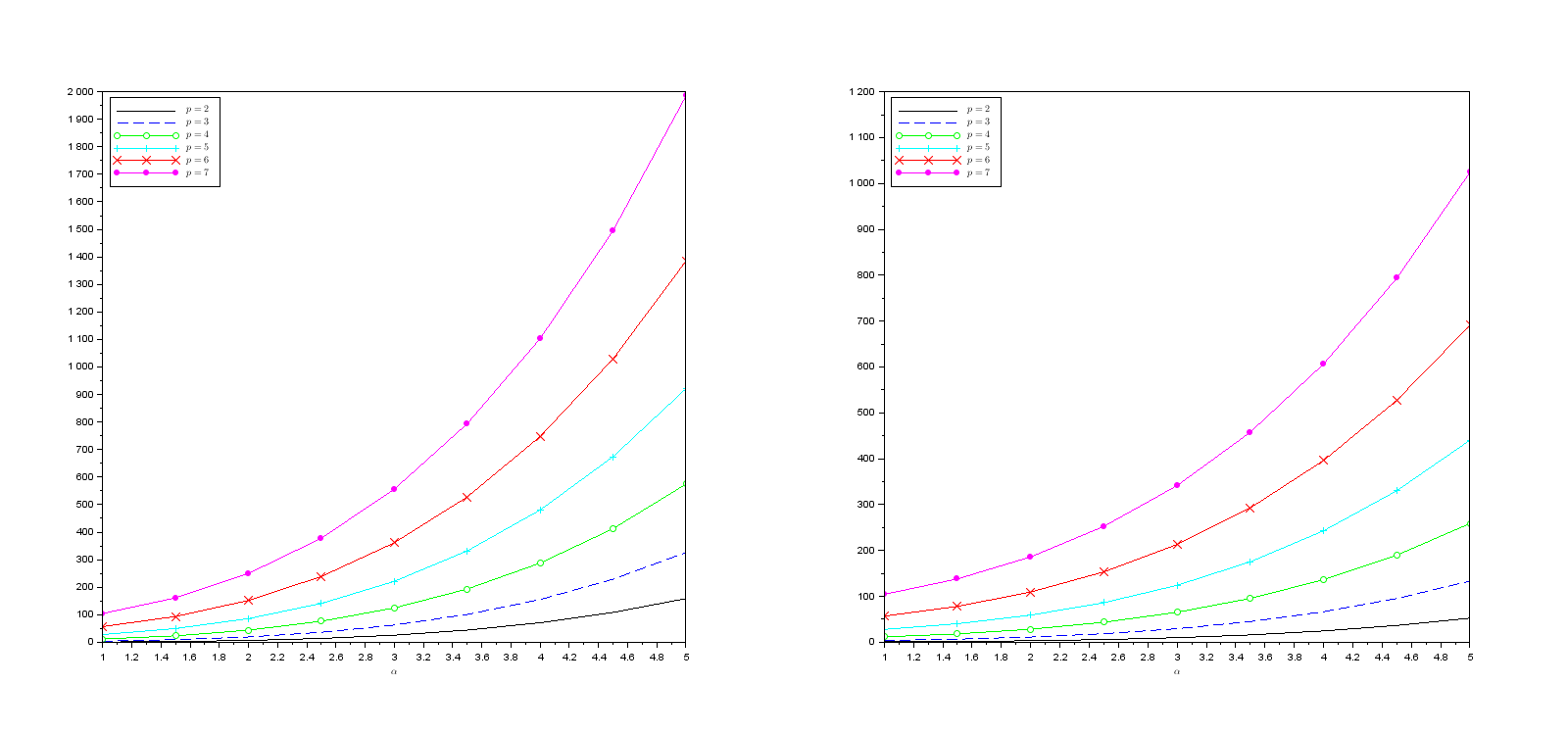}
\caption{Linear model defined in \eqref{def:linear_model}. The difference between the limiting variances with respect to $X_1$ (left panel) and to $X_2$ (right panel) are plotted. As explained, the Pick-Freeze estimation procedure has been weighted by $(p+1)$ to have a fair comparison. The number of variables involved in the model varies from $p=2$ to $p=7$.}
\label{fig:toy_diff}
\end{figure}

\begin{rem}
Observe that a more precise comparison should consists in comparing (via definite-positiveness) the limiting covariance-variance matrices involving both the limiting variances and the limiting covariances.
If it is straightforward to compute the covariance terms for the Pick-Freeze methodology: for $i=2,\ldots,p$, 
\begin{align*}
\Cov(YY^1,YY^i)&= \frac{1}{24}\alpha^4
+\frac{1}{12}m_{1,p-1}\alpha^3+\Bigl(\frac{7}{144}+\frac 14 v_{p-1}+
\frac 16\Bigl(m_{1,p-1}+\frac 12\Bigr)^2\Bigr)\alpha^2\\
&\qquad +\Bigl(\frac 18+\frac{1}{12}m_{1,p-1}+\frac 12 v_{p-1}+v_{p-1}m_{1,p-1}\Bigr)\alpha
+v_{p-1}\Bigl(m_{1,p-1}+\frac 12\Bigr)^2,
\end{align*}
it is much more tricky to deal with the rank-based procedure. Indeed, to do so a joint CLT is required for the vector of all $p$ first-order Sobol' indices whose proof is not a direct generalization of the proof of Theorem \ref{th:tcl}. Such an extension will be done in a forthcoming paper. 
\end{rem}

\paragraph*{Asymptotic efficiency}

The two previous procedures do not rely on the same design of experiment so that it is not possible to determine which one is the more efficient in the sense of \cite[Section 25]{van2000asymptotic}.

By \cite[Proposition 2.5]{pickfreeze}, the sequence of estimators $(T^1_n)_n$ is asymptotically efficient to estimate $S^1$ when the distribution $P$ of $(Y,Y^1)$ belongs to $\mathcal P$, the set of all c.d.f.\ of exchangeable random vectors in $L^2(\R^2)$.

Using a unique $n$-sample, one may compare the rank-based estimators introduced in this paper and the procedure involving the estimators $\widehat T_n$ defined in \cite[page 11]{da2008efficient}. 
Such estimator is particularly tricky to compute and not easily tractable in practice. More precisely, the initial $n$-sample is split into two samples of sizes $n_1$ and $n_2=n-n_1$. The first sample is dedicated  to the estimation of the joint density of $(X,Y)$ while the second one is used to compute a Monte-Carlo estimation of the integral involved in the quantity of interest. 
In a work under progress \cite{GKLPdV21}, another estimator based on kernels and the same design of experiment is proposed. This estimator is more tractable in practice. 

By \cite[Theorems 3.4 and 3.5]{da2008efficient}, the sequence of estimators $(\widehat T_n)_n$ is asymptotically efficient to estimate $\E[\E[Y\vert X]^2]$ leading to an asymptotically efficient sequence of estimators of $S^1$. The proof of the following proposition has been postponed in Appendix \ref{app:ae}.

\begin{prop}\label{prop:ae}
Consider the sequence of estimators $\widehat T_n$ introduced in \cite[page 11]{da2008efficient}. Assume that the joint distribution $P$ of $(X,Y)$ is absolutely continuous with respect to the product probability $P_X\otimes P_Y$, namely $P(dx,dy)=f(x,y)P_X(dx) P_Y(dy)$. Then the sequence $(R^1_n)_n$
\begin{align*}
R^1_n=\frac{\widehat T_n - \left(\frac{1}{n} \sum_{i=1}^n Y_i\right)^2}{\frac{1}{n} \sum_{i=1}^n Y_i^2-\left(\frac{1}{n} \sum_{i=1}^n Y_i\right)^2}
\end{align*}  
is asymptotically efficient in estimating $S^1$. In addition, its (minimal) variance $\sigma_{\min}^2$ is 
\begin{align*}
\sigma_{\min}^2\defeq \frac{1}{\Var(Y)^2}\Var\left(2\E[Y](1-S^1)Y+S^1Y^2+\E[Y\vert X](\E[Y\vert X]-2Y)\right).
\end{align*}
\end{prop}

Thus we are interested in the comparison of $\sigma_{\min}^2$ and $\sigma^2$ given in Theorem \ref{th:tcl}. Let us consider again the example of the linear model \eqref{def:linear_model} introduced in the previous paragraph.

\textbf{Example (continued)}. 
 We consider the model defined in \eqref{def:linear_model}.
As done in the previous paragraph, we only compare $V^1_{\text{Eff}}\defeq \Var(\E[Y\vert X_1](2Y-\E[Y\vert  X_1]))$ to $\Sigma_B^{1,1}+\Sigma_C^{1,1}$ and $V^i_{\text{Eff}}\defeq \Var(\E[Y\vert X_i](2Y-\E[Y\vert  X_i]))$ to $\Sigma_B^{i,i}+\Sigma_C^{i,i}$ for $i=2,\ldots,p$. After some trivial computations, one gets 
\begin{align*}
V^1_{\text{Eff}}&=\frac{4}{45}\alpha^4+\frac{1}{3}m_{1,p}\alpha^3+\frac{1}{3}\Bigl(4v_p+m_{1,p}^2\Bigr)\alpha^2+4m_{1,p}v_p\alpha+4 v_p m_{1,p}^2,\\
V^i_{\text{Eff}}&=\frac{4}{45}+\frac{1}{3}m_{1,p,\alpha}+\frac{1}{3}\Bigl(4v_{p,\alpha}+m_{1,p,\alpha}^2\Bigr)+4m_{1,p,\alpha}v_{p,\alpha}+4 v_{p,\alpha} m_{1,p,\alpha}^2.
\end{align*}

We compare these limiting variances in Figure \ref{fig:toy_eff}. We observe that the limiting variances obtained with the rank methodology do not differ much from the efficient variances. 

\begin{figure}[h!]
\includegraphics[scale=0.48]{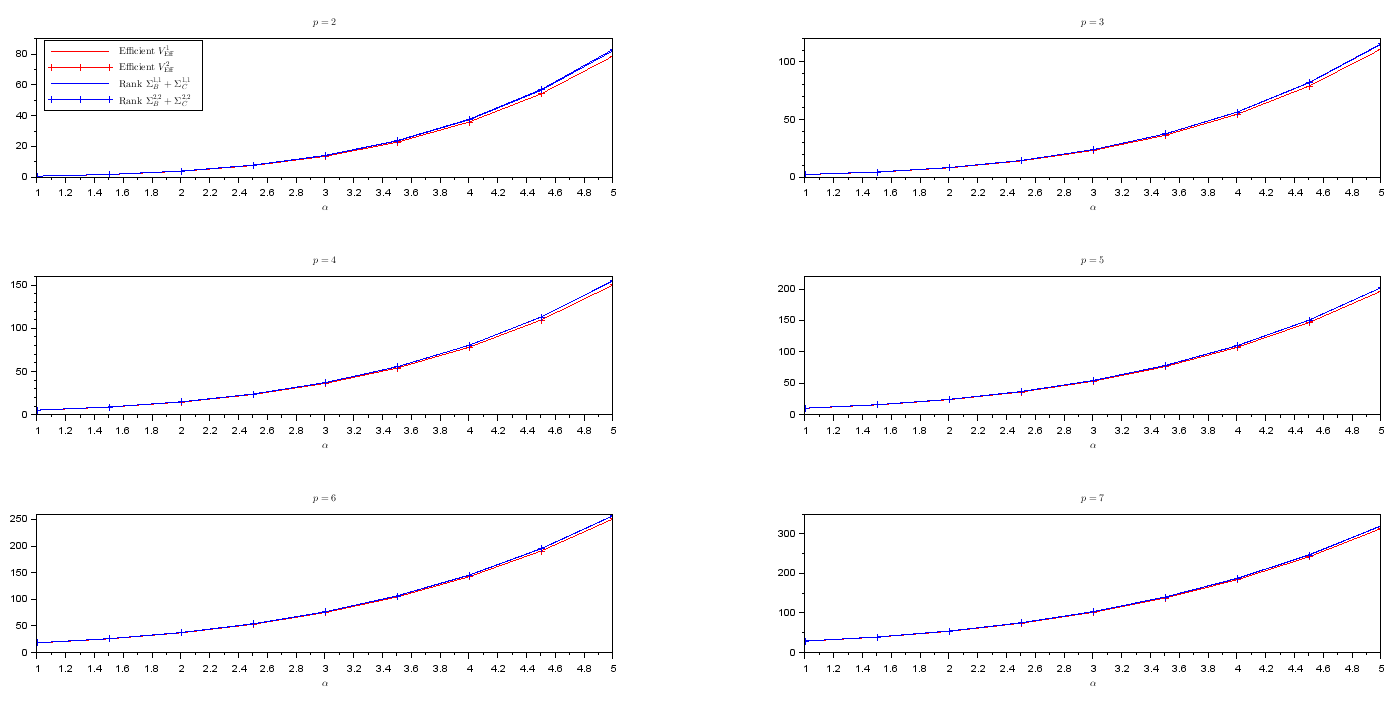}
\caption{Linear model defined in \eqref{def:linear_model}. The limiting variances with respect to $X_1$ (plain lines) and to $X_2$ (plain lines with +) are plotted. The rank-based estimation procedure is represented in blue while the efficient variances are represented in red. The number of variables involved in the model varies from $p=2$ to $p=7$.}
\label{fig:toy_eff}
\end{figure}

\subsection{Recovering other classical indices}\label{sec:known_indices}

In \cite{FGM2017}, the authors considered computer codes of the form  \eqref{def:model} valued on a compact Riemannian manifold. In this framework, they proposed a sensitivity index in the flavour of the Cramé-von-Mises index and they used the Pick-Freeze scheme to provide a consistent estimator. The authors of \cite{GKLM19} extend the previous indices to the context of general metric spaces and propose U-statistics-based estimators improving the classical Pick-Freeze procedure. In light of Section \ref{sec:gene}, one may introduce a novel estimation of the indices introduced in \cite{GKLM19} requiring a unique $n$-sample. The reader is referred to \cite{FKL21} for more details on the procedure.

Following \cite{Owen12,ODC13}, extensions to Sobol' indices are obtained by replacing  their numerator  by  higher-order moments. 
In \cite{GKL18}, the authors construct a Pick-Freeze  estimator for such extensions. One again, we are now able to propose another estimation 
scheme based on a unique $n$-sample. The reader is referred to \cite{GKL21} for the generalization of Lemma \ref{lem:pfChatt} and the corresponding asymptotic study.

\section{Numerical experiments}\label{sec:exnum}

\subsection{Numerical comparison on the Sobol' \texorpdfstring{$g$}{f}-function: conventional Pick-Freeze estimators vs rank estimators}

In this section, we compare the performances of both estimation procedures on an  analytic function: the so-called  Sobol' $g$-function, that is defined by
\begin{equation}\label{eq:gfct}
g(X_1,\ldots,X_{p})=\prod_{i=1}^{p} \frac{\abs{4X_i-2}+a_i}{1+a_i},
\end{equation}  
where $(a_i)_{i\in \N}$ is a sequence of real numbers and the $X_i$'s are i.i.d.\  random variables uniformly distributed on $[0,1]$. In this setting, one may easily compute the exact expression of the first-order Sobol' indices: 
\[
S^i=\frac{(1+a_i^2)^{-1}/3}{3^{-p}\prod_{i=1}^p (1+a_i^2)^{-1}-1}.
\]
As expected, the lower the
coefficient $a_i$, the more significant the variable $X_i$. In the sequel, we simply fix $a_i=i$. 
Due to its complexity (non-linear and non-monotonic correlations) and the analytical expression of the Sobol' indices, the Sobol' $g$-function is a classical test example commonly used in GSA (see e.g.\ \cite{saltelli-sensitivity}).

\paragraph*{Convergence as the sample size increases} 
In Figure \ref{fig:jouet_N_increases}, we compare the estimations of the six first-order Sobol' indices given by both methods ($p=6$). In the Pick-Freeze estimations given by \eqref{def:Tn}, several sizes of sample $N$ have been considered: $N=100$, 500, 1000, 5000, 10000, 50000, 100000, and 500000. The Pick-Freeze procedure requires $(p+1)=7$ samples of size $N$. To have a fair comparison, the sample sizes considered in the estimation of $\xi_n^{\text{Sobol'}}$ are $n=(p+1)N=7 N$. Both methods converge and give precise results for large sample sizes. 

\begin{figure}[h!]
\includegraphics[scale=0.48]{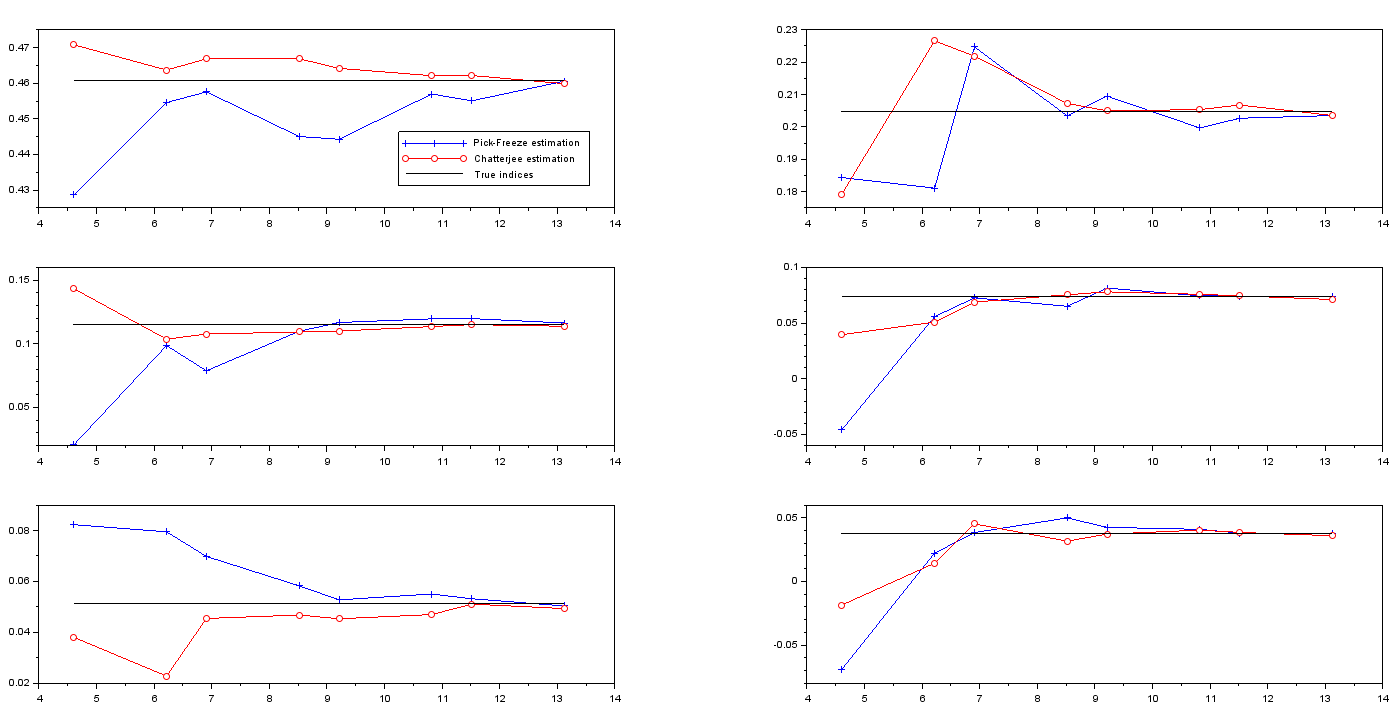}
\caption{The Sobol' $g$-function model \eqref{eq:gfct}. Convergence of both methods when $N$ increases. The sixth first-order Sobol' indices have been represented from left  to right and up to bottom. Several sample sizes have been considered: $N=100$, 500, 1000, 5000, 10000, 50000, 100000, and 500000 for the Pick-Freeze estimation procedure (in blue) and correspondingly $(p+1)N$ for the rank estimation procedure (in red). The true indices are displayed in black plain line. The $x$-axis is in log. scale. }
\label{fig:jouet_N_increases}
\end{figure}

\paragraph*{Comparison of the mean square errors}
We now compare the efficiency of both methods at a fixed sample size. In that view, we assume that only $n=700$ calls of the computer code $f$ are allowed to estimate the six first-order Sobol' indices. We repeat the estimation procedure 500 times. The boxplot of the mean square errors for the estimation of the first-order Sobol' index $S^1$ with respect to $X_1$ has been represented in Figure \ref{fig:boxplot}. We observe that, for a fixed sample size $n=700$ (corresponding to a Pick-Freeze sample size $N=100$), the rank estimation procedure performs much better than the Pick-Freeze method with significantly lower mean errors. The same behavior can be observed for all the first Sobol' indices as can be seen in Table \ref{tab:modlin} that provides some characteristics of the mean squares errors.

\begin{figure}
\begin{tabular}{cc}
\includegraphics[scale=.50]{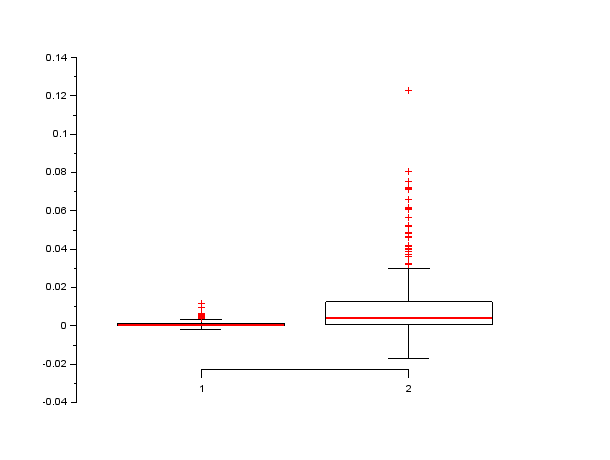}
\end{tabular}
\caption{The Sobol' $g$-function model \eqref{eq:gfct}. Boxplot of the mean square errors of the estimation of $S^1$ with a fixed sample size and 500 replications. The results of the rank methodology with $n=700$ are provided in the left panel. The results of the Pick-Freeze estimation procedure with $N=100$ are provided in the right panel.}
\label{fig:boxplot}
\end{figure}

\begin{table}[ht]
\begin{small}
\begin{tabular}{c|c|c|c|c|c|c}
\cline{2-7}
\multicolumn{1}{c}{} &\multicolumn{3}{c|}{Pick-Freeze}&\multicolumn{3}{c}{Rank}\\
\cline{2-7}
\multicolumn{1}{c}{} & Mean & Median & Stdev & Mean & Median & Stdev \\
\hline
mse $S^1$ & 0.0095548 & 0.0039458 & 0.0145033 & 0.0010218 & 0.0004498 & 0.0013999\\ 
mse $S^2$ & 0.0105727 & 0.0046104 & 0.0148873 & 0.0017314 & 0.0006870 & 0.0027436\\
mse $S^3$ & 0.0101785 & 0.0041789 & 0.0143846 & 0.0016667 & 0.0006409 & 0.0024392\\
mse $S^4$ & 0.0105463 & 0.0047284 & 0.0178064 & 0.0018522 & 0.0008126 & 0.0025296\\  
mse $S^5$ & 0.0097979 & 0.0042995 & 0.0135533 & 0.0016285 & 0.0006855 & 0.0024264\\  
mse $S^6$ &   0.0096109 & 0.0046822 & 0.0134822 & 0.0015590 & 0.0007080  & 0.0021333\\ 
\hline
\end{tabular}
\caption{\label{tab:modlin}The Sobol' $g$-function model \eqref{eq:gfct}. Characteristics of the mean square errors for the estimation of the six first-order Sobol' indices with a fixed sample size and 500 replications. In the rank methodology, the sample size is $n=700$ while in the Pick-Freeze estimation procedure, it is $N=100$.}
\end{small}
\end{table}

\paragraph*{Performances for small sample sizes or for large number of input variables} 
As expected, we can observe in Table \ref{tab:modlin2_mse} that the rank estimation procedure proceeds much better than the Pick-Freeze methodology for small sample sizes. Similarly, if the number of input variables increases drastically, we can observe the same behavior as can be seen in Figure \ref{fig:var}. In that case, we consider the model \eqref{eq:gfct}
for several values of $p$: 6, 10, 15, 20, 30, 40, and $50$.

\begin{table}[ht]
\begin{small}
\begin{tabular}{c|c|c|c|c|c|c}
\cline{2-7}
\multicolumn{1}{c}{} &\multicolumn{3}{c|}{Pick-Freeze}&\multicolumn{3}{c}{Rank}\\
\cline{2-7}
\multicolumn{1}{c}{} & $N=10$ & $N=50$ & $N=100$ & $n=70$ & $n=350$ & $n=700$ \\
\hline
mse $S^1$ & 0.1128686 & 0.0172275 & 0.0095548 & 0.0116790 & 0.0022941 & 0.0010218  \\
mse $S^2$ & 0.1509575 & 0.0223196 & 0.0105727 & 0.0177522 & 0.0033719 & 0.0017314\\  
mse $S^3$ & 0.1469124 & 0.0220015 & 0.0101785 & 0.0175517 & 0.0032474 & 0.0016667\\  
mse $S^4$ & 0.1591130 & 0.0196357 & 0.0105463 & 0.0159360 & 0.0033948 & 0.0018522\\  
mse $S^5$ & 0.1646339 & 0.0240353 & 0.0097979 & 0.0158563 & 0.0032230 & 0.0016285\\  
mse $S^6$ & 0.1466408 & 0.0217638 & 0.0096109 & 0.0166701 & 0.0029653 & 0.0015590\\ 
\hline
\end{tabular}
\caption{\label{tab:modlin2_mse} The Sobol' $g$-function model \eqref{eq:gfct}. Mean squares errors of the estimation of the six first-order Sobol' indices with small sample sizes and with both methods.}
\end{small}
\end{table}

\begin{figure}[ht]
\begin{tabular}{cc}
\includegraphics[scale=.48]{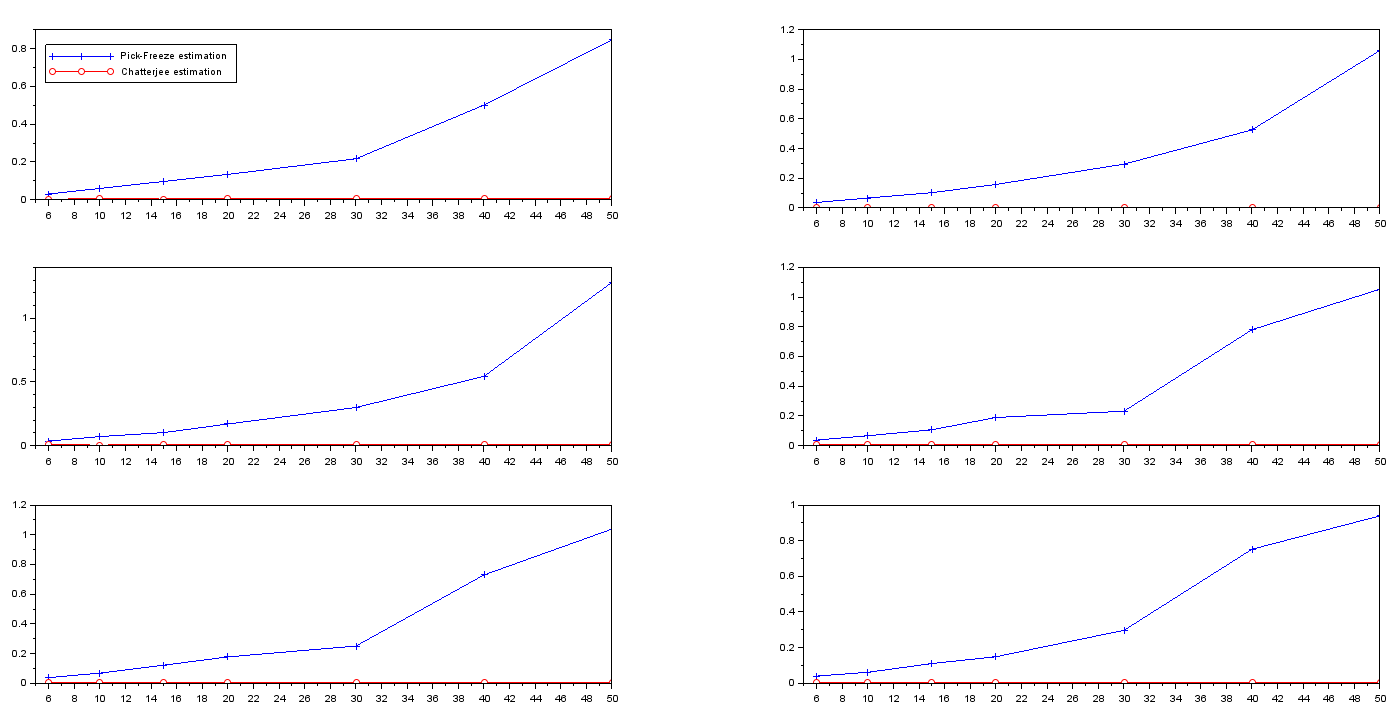}
\end{tabular}
\caption{The Sobol' $g$-function model \eqref{eq:gfct}. Mean square errors of the estimation of the six first-order Sobol' indices with respect to the number of input variables with a fixed sample size and 500 replications. We consider the sample sizes $n=200$ in the rank methodology (in red) and  $N=n/(p+1)$ in the Pick-Freeze procedure (in blue). The number of input variables considered are $p=6,10,15,20,30,40$, and $50$.}
\label{fig:var}
\end{figure}

\subsection{An application in biology}\label{ssec:biolo}

Here, we illustrate the nature and the performance of the Cram\'er-von-Mises indices and their corresponding rank estimators as a screening mechanism for high-dimensional problems. To do so, 
we consider the neurovascular coupling model from \cite{hgd}. Mathematically, this corresponds to the following differential-algebraic equation (DAE) system
\begin{align}
\frac{dW}{dt} &= G(W,Z,X), \quad
0 = H(W,Z,X), \label{algnvu}
\end{align}
where $W = (W_1, \dots, W_N)$ and $Z = (Z_1, \dots, Z_M)$ correspond respectively to the differential and algebraic state variables of the models. The variables $X = (X_1 , \dots, X_p)$ correspond to the uncertain parameters of the model. Our quantity of interest corresponds to the time average over $[0,T]$ of $W^*$ (which is one of the differential state variables $W_1$, ..., $W_N$), i.e.
\begin{align}
Y = \frac 1T \int_0^\top W^*(t) \, dt. \label{QoI}
\end{align}
As above, we regard $Y$ as a function of the unknown parameters, i.e., $Y = f(X_1, \dots, X_p)$. In our implementation, the values of $W^*$ are obtained by solving the above DAE system (Equation \eqref{algnvu}) by the MATLAB routine ode15s (it  can be checked that  \eqref{algnvu} form an index one system). Further,  in the current example, $N=67$ and $p = 160$ and the distributions of most of the $X_i$'s are uniform and allowed to vary $\pm 10\%$ from nominal values (see \cite{hgd} for additional details). 

We compare the results from the rank estimators as described above to those resulting from the  linear regression
\[
f(X_1, \dots , X_{160}) \approx \lambda _0 + \sum_{j=1}^{160} \lambda _j X_j. 
\]
As shown in \cite{hgd}, the above approximation performs well for the considered QoI. We assign to each variable $X_1, \dots, X_{160}$  a relative importance $L_j$ where
\[
L_j = \frac{|\lambda_j|}{\sum_{\ell=1}^{160} |\lambda_{\ell}|}, \qquad j = 1, \dots, 160.
\]

Figure~\ref{fig:nvu} displays the results. Both screening approaches identify the same to three influential parameters. More parameters are identified as being non-influential through the linear regression approach than using the Cram\'er-von-Mises indices.

\begin{figure}[h!]
\centering
\begin{tabular}{cc}
\includegraphics[width=0.7\textwidth]{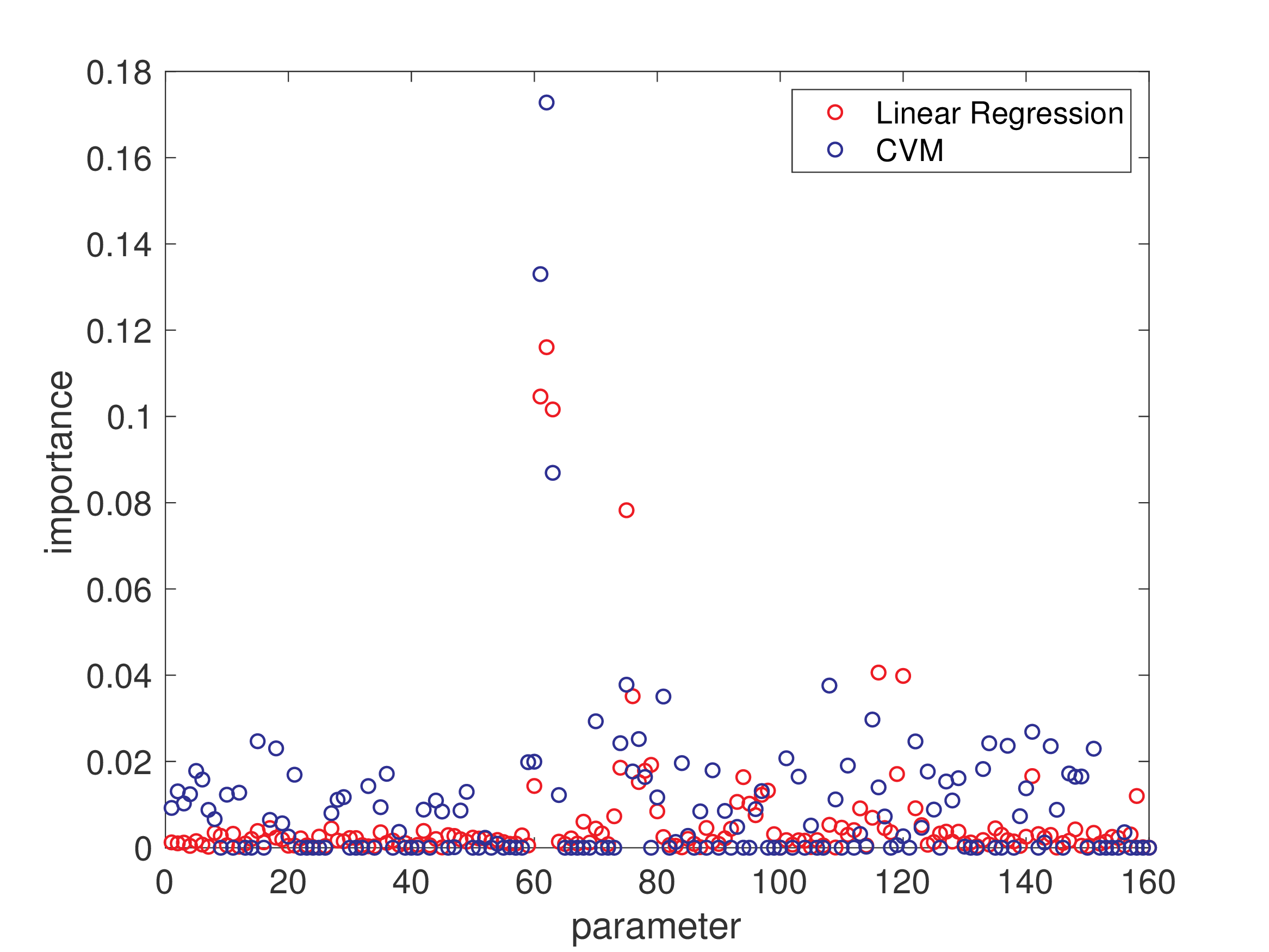}
\end{tabular}
\caption{Rank estimators corresponding to the Cram\'er-von-Mises indices as a screening mechanics for the DAE system given by \eqref{algnvu} and \eqref{algnvu}.}
\label{fig:nvu}
\end{figure}

\section{Conclusion}\label{sec:conclu}

In this paper, we explain how to use the estimator proposed by Chatterjee in \cite{Chatterjee2019}  to provide a very nice and mighty procedure to estimate both all the first-order Sobol' indices and the so-called Cram\'er-von-Mises indices \cite{GKL18} at a small cost (only $n$ calls of the computer code). We emphasize on the fact that this estimation procedure requires a unique sample contrary to the Pick-Freeze procedure based on a particular design of experiment, the size of which is 2n when estimating a single index and increases with the number of indices to estimate. 
We also extend Chatterjee's method to estimate more general quantities. 
Furthermore, we show a CLT for our estimations of Sobol' indices.
As examples,
we consider two indices already introduced in sensitivity analysis: the indices adapted to output valued in general metric spaces defined in \cite{GKLM19} and the higher-moment indices \cite{Owen12,ODC13}. A general CLT will be established soon in \cite{GKL21}.

\textbf{Acknowledgment}. We warmly thank Robin Morillo for the numerical study provided in Section \ref{ssec:biolo}. Moreover, we deeply thank the anonymous referee of the early version of our paper who pushed us to prove the CLT. We also gratefully thank the anonymous reviewer of the current version of this paper for his comments, critics and advises, which
greatly helped us to improve the manuscript. 

Support from the ANR-3IA Artificial and Natural Intelligence Toulouse Institute is gratefully acknowledged. This work was also supported by the National Science Foundation under grant DMS-1745654.

\appendix

\section{Proof of the consistency}\label{app:cons}

\begin{proof}[Proof of Lemma \ref{lem:pfChatt}]
Since $\tau_n$ has no fix point, and using the measurability of $\tau_n$ and the independence, we have
\begin{align*}
\E&\left[g(Y_j)h(Y_{\tau_n(j)})|\mathcal F_n\right]
=\E\Bigl[g(Y_j)\sum_{\substack{l=1,\\l\neq j}}^n h(Y_{l})\ind_{\{\tau_n(j)=l\}}|\mathcal F_n\Bigr]
=\sum_{\substack{l=1,\\l\neq j}}^n\ind_{\{\tau_n(j)=l\}}\E\Bigl[g(Y_j) h(Y_{l})|\mathcal F_n\Bigr]\\
&=\sum_{\substack{l=1,\\l\neq j}}^n\ind_{\{\tau_n(j)=l\}}\E\Bigl[g(Y_j) |\mathcal F_n\Bigr] \E \Bigl[h(Y_{l})|\mathcal F_n\Bigr] =\E\bigl[g(Y_j) |V_j\bigr]\sum_{\substack{l=1,\\l\neq j}}^n\ind_{\{\tau_n(j)=l\}} \E \bigl[h(Y_{l})|V_l\bigr]\\
&=\Psi_{V_j}(g)\sum_{\substack{l=1,\\l\neq j}}^n\ind_{\{\tau_n(j)=l\}} \Psi_{V_l}(h)
=\Psi_{V_j}(g)\Psi_{V_{\tau_n(j)}}(h). \qedhere
\end{align*}
\end{proof}

\begin{proof}[Proof of Proposition \ref{prop:cv}]
We follow the steps of the proof of Corollary 7.12 in \cite{Chatterjee2019}. Our proof is significantly simpler since $\tau_n$ is assumed to have no fix points and $V$ is continuous so that there are no ties in the sample. To simplify the notation, we  denote $\chi_n(V,Y;g,h)$ and $\chi(V,Y;g,h)$ by $\chi_n$ and $\chi$ respectively. 

We first prove that, for any measurable function $\varphi$, 
\begin{align}\label{eq:cv_lusin}
\varphi(V_1)-\varphi(V_{\tau_n(1)})\to 0
\end{align} in probability  as $n\to \infty$. Let $\varepsilon>0$. By the special case of Lusin's theorem (see \cite[Lemma 7.5]{Chatterjee2019}), there exists a compactly supported continuous function $\tilde \varphi\colon \R\to\R$ such that $\P(\{x;\, \varphi(x)\neq \tilde \varphi(x)\})<\varepsilon$, where $\P$ stands for the distribution of $V$. 
Then for any $\delta>0$, 
\begin{align}
\P\Big(\abs{\varphi(V_1)-\varphi(V_{\tau_n(1)})}&>\delta\Big)
\leqslant \P\left(\abs{\tilde \varphi(V_1)-\tilde \varphi(V_{\tau_n(1)})}>\delta\right)\nonumber\\
&+\P\left(\varphi(V_1)\neq \tilde \varphi(V_1))+\P(\varphi(V_{\tau_n(1)})\neq \tilde \varphi(V_{\tau_n(1)})\right).\label{ineq:cv_lusin}
\end{align}
By continuity of $\tilde \varphi$ and since $V_{\tau_n(1)}\to V_1$ as $n\to \infty$ with probability one, the first term in the right hand side of \eqref{ineq:cv_lusin} converges to 0 as $n\to\infty$. 
By construction of $\tilde \varphi$, the second term is lower than $\varepsilon$. Turning to the third one, we have thus 
\begin{align*}
\E&[\varphi(V_{\tau_n(1)})]=\frac 1n \sum_{j=1}^n \E[\varphi(V_{\tau_n(j)})] =  \frac 1n \sum_{j=1}^n \sum_{\substack{l=1\\ l\neq j}}^n \E[\varphi(V_l)\ind_{\{\tau_n(j)=l\}}]\\
&=  \frac 1n \sum_{l=1}^n \sum_{\substack{j=1\\ j\neq l}}^n \E[\varphi(V_l)\ind_{\{\tau_n(j)=l\}}] =  \frac 1n \sum_{l=1}^n \E[\varphi(V_l)  \sum_{\substack{j=1\\ j\neq l}}^n\ind_{\{\tau_n(j)=l\}}]=  \frac 1n \sum_{l=1}^n \E[\varphi(V_l)] =\E[\varphi(V_1)]
\end{align*}
where we have used the fact that  $\tau_n$ has no fix point, $V_{\tau_n(i)}\overset{\mathcal L}{=} V_{\tau_n(j)}$ for any $i$ and $j=1,\dots,n$, and the $V_i$'s have no ties. This yields
\begin{align*}
\P(\varphi(V_{\tau_n(1)})\neq \tilde \varphi(V_{\tau_n(1)}))=\P(\varphi(V_1)\neq \tilde \varphi(V_1))<\varepsilon,
\end{align*}
and, since $\varepsilon$ and $\delta$ are arbitrary, \eqref{eq:cv_lusin} is therefore proved. 
Now, since $x\mapsto \Psi_x$ is a measurable and bounded function and applying \eqref{eq:cv_lusin}, we  have
\begin{align}
\left\{\begin{array}{ll} \Psi_{V_1}(g)-\Psi_{V_{\tau_n(1)}}(g) &\to 0, \\ 
\Psi_{V_1}(h)-\Psi_{V_{\tau_n(1)}}(h) &\to 0,\end{array}
\right. \quad\mbox{in probability as $n\to \infty$.}  \label{eq:cv_psi}
\end{align}
Lemma \ref{lem:pfChatt} and the dominated convergence theorem lead to
\begin{footnotesize}
\begin{align}
\E[\chi_n] &= \frac 1n \sum_{j=1}^n \E[g(Y_j)h(Y_{\tau_n(j)})]
= \E[g(Y_1)h(Y_{\tau_n(1)})]
=  \E[\Psi_{V_1}(g)\Psi_{V_{\tau_n(1)}}(h)]
\to \E[\Psi_{V}(g)\Psi_{V}(h)]= \chi\label{eq:cv_espe}
\end{align}
\end{footnotesize}
where we have taken into account the fact that  $\Psi_V(g)$ and $\Psi_V(h)$ are bounded (due to the boundedness of $g$ and $h$) and used \eqref{eq:cv_psi}.

The last step of the proof consists in comparing $\chi_n$  with $\E[\chi_n] $ using Mc Diarmid's concentration inequality  \cite{mcdiarmid1989method}.
Sharper constants can be obtained in Mc Diarmid's inequality by using the inequalities from \cite{boucheron2013concentration,boucheron2009concentration}. 
As we are interested in asymptotic results the accuracy of the constant has no impact on the result. 
Following the same lines as in the proof of \cite[Lemma 7.11]{Chatterjee2019}, Mc Diarmid's concentration inequality  in \cite{mcdiarmid1989method} then implies
\begin{align}\label{eq:concentration}
\P(\abs{\chi_n-\E[\chi_n]}\geqslant t)\leqslant 2\exp \{-2n^2t^2/C^2\},
\end{align}
where $C$ is a universal constant and we  conclude the proof by  combining \eqref{eq:cv_espe} and \eqref{eq:concentration}.
\end{proof}

\section{Proof of the asymtotic normality 
}\label{app:asymp_norm}

%
%

\textbf{Framework and goal} 
We consider  the model defined in \eqref{def:model} that can be rewritten as 
$Y=f(X,W)$
where $X=X_1$ and $W=(X_2,\ldots,X_p)$ are two independent inputs of the numerical code $f$ that is assumed to be bounded.

The random variables $X$ and $W$ are defined on a product space $\Omega=\Omega_X\times \Omega_W$; so that for any $\omega\in \Omega$, there exists $\omega_X\in \Omega_X$ and $\omega_W\in \Omega_W$ and we have $(X,W)(\omega)=(X(\omega_X),W(\omega_W))$. Further, we consider $\pi_W$ the projection on $\Omega_W$ and the product measure $\P=\P_X\otimes \P_W=\mathcal L_X\otimes \mathcal L_W$, where $\mathcal L_X$ is the distribution of $X$ and $\mathcal L_W$ is the distribution of $W$.  Naturally, $\P_W=\P\circ \pi_W^{-1}$.

 We aim to prove a CLT  for the estimator $\xi_n^{\text{Sobol'}}(X,Y)$ of the classical first-order Sobol' index with respect to $X$ given by \eqref{eq:defsob}, the estimator of which defined in \eqref{eq:coeffChatterjee_version_Sobol'}  is given by
\begin{align*}
\xi_n^{\text{Sobol'}}(X_1,Y) = \frac{\frac 1n\sum_{j=1}^n Y_jY_{N(j)}-\left(\frac 1n \sum_{j=1}^n Y_j\right)^2}{\frac 1n \sum_{j=1}^n Y_j^2-\left(\frac 1n \sum_{j=1}^n Y_j\right)^2}
\end{align*}
where $N$ is defined in \eqref{def:Nnous}. Notice that the denominator is reduced to the empirical variance of $Y$. As explained in Section \ref{ssec:estim_cvm_chatterjee}, we denote by $Y_{(j)}$ the output associated to $X_{(j)}$ where $X_{(j)}$ stands for the $j$-th order statistics of $(X_1,\ldots,X_n)$. Then observing that 
\begin{align*}
\sum_{j=1}^n Y_jY_{N(j)}=\sum_{j=1}^n Y_{(j)}Y_{(j+1)}\eqdef\sum_{j=1}^n Y_{\sigma_n(j)}Y_{\sigma_n(j+1)}
\end{align*}
where, to avoid any confusion, $\sigma_n$  
stands for the permutation that rearranges the sample $(X_1,\ldots,X_n)$, the estimator $\xi_n^{\text{Sobol'}}(X_1,Y)$ can be written as 
\begin{align}\label{eq:coeffChatterjee_version_Sobol'2}
\xi_n^{\text{Sobol'}}(X_1,Y) = \frac{\frac 1n\sum_{j=1}^{n-1} Y_{\sigma_n(j)}Y_{\sigma_n(j+1)}-\left(\frac 1n \sum_{j=1}^n Y_{\sigma_n(j)}\right)^2}{\frac 1n \sum_{j=1}^n Y_{\sigma_n(j)}^2-\left(\frac 1n \sum_{j=1}^n Y_{\sigma_n(j)}\right)^2}.
\end{align}

\subsection{Proof of Theorem \ref{th:tcl}} \label{sec:proof_th}

The proof will proceed as follows. First, in view of \eqref{eq:coeffChatterjee_version_Sobol'2}, we prove a CLT for 
\begin{equation*}
\left(\frac 1n\sum_{j=1}^{n-1}\ysnj\ysnjpu,\frac 1n\sum_{j=1}^{n}\ysnj,\frac 1n\sum_{j=1}^{n}\ysnj^2 \right).
\end{equation*}
that amounts to prove a CLT for 
\begin{equation*}
\left(\frac 1n\sum_{j=1}^{n-1}\ysnj\ysnjpu,\frac 1n\sum_{j=1}^{n-1}\ysnj,\frac 1n\sum_{j=1}^{n-1}\ysnj^2 \right),
\end{equation*}
since $f$ is bounded. 
Secondly, we use the so-called delta method \cite[Theorem 3.1]{van2000asymptotic} 
to conclude to Theorem \ref{th:tcl}.

It is worth noticing that the permutation on the $W$'s do not affect the result as  seen in the sequel. For $j=1,\ldots n-1$, introducing 

\begin{align}\label{def:Wnj}
\Delta_{n,j}\defeq f\lp X_{\sn(j)},W_{j}\rp-f\lp \frac{j}{n+1},W_{j}\rp, \quad   W_{n,j}\defeq \bigl( \frac{j}{n+1},W_{j}\bigr)
\end{align} 
leads to $
\ysnj = f\lp X_{\sn(j)},W_{\sn(j)}\rp
\egloi f\lp X_{\sn(j)},W_j\rp= \Delta_{n,j} +f\lp   W_{n,j} \rp $
and
\begin{align*}
\ysnj& \ysnjpu = f\lp X_{\sn(j)},W_{\sn(j)}\rp f\lp X_{\sn(j+1)},W_{\sn(j+1)}\rp\\
&\egloi f\lp X_{\sn(j)},W_j\rp f\lp X_{\sn(j+1)},W_{j+1}\rp\\
&=\Bigl (  f\lp   W_{n,j}\rp +\Delta_{n,j} \Bigr) \Bigl (  f\lp   W_{n,j+1}\rp +\Delta_{n,j+1} \Bigr)\\
&=   f\lp   W_{n,j}\rp f\lp   W_{n,j+1}\rp +  \Delta_{n,j} f\lp   W_{n,j+1}\rp+ \Delta_{n,j+1}f\lp   W_{n,j}\rp+ \Delta_{n,j}\Delta_{n,j+1}.
\end{align*}
Thus we are led to establish a CLT for
\begin{align}\label{def:Zn}
Z_n=\frac 1n \sum_{j=1}^{n-1} &\begin{pmatrix} 
 f(   W_{n,j}) f(   W_{n,j+1}) +   \Delta_{n,j} f\lp   W_{n,j+1}\rp+ \Delta_{n,j+1}f\lp   W_{n,j}\rp
 + \Delta_{n,j}\Delta_{n,j+1}\\
 f(  W_{n,j})+\Delta_{n,j}\\
 \bigl( f(   W_{n,j})+\Delta_{n,j}\bigr)^2
 \end{pmatrix}.
\end{align}

Let us discard the negligible terms in the CLT for $Z_n$. In that view, noticing that 
\[
\E\left[X_{\sn(j)}\right]=\frac{j}{n+1}
\quad \text{and} \quad
\Var( X_{\sn(j)})=\frac{j(n-j+1)}{(n+1)^2(n+2)}=\E\left[\lp X_{\sn(j)}-\frac{j}{n+1}\rp^2\right]\leqslant \frac{4}{n+2},
\]
we first establish 
\begin{align}\label{eq:X_behav}
 X_{\sn(j)}-\frac{j}{n+1}=O_{\P}\lp \frac{1}{\sqrt n}\rp.
\end{align}
As explained below, \eqref{eq:X_behav} will imply
\begin{align}\label{eq:delta}
%
\frac 1n \sum_{j=1}^{n-1} \Delta_{n,j}^2=O_{\P}\lp\frac {1}{ n}\rp \quad \text{and} \quad
\frac 1n \sum_{j=1}^{n-1} \Delta_{n,j}\Delta_{n,j+1}=O_{\P}\lp\frac {1}{ n}\rp.
\end{align}

First of all, we expand  $\Delta_{n,j}$ (resp. $\Delta_{n,j+1}$) using the Taylor-Lagrange formula, for any  $j=1,\ldots n-1$ and we obtain
\begin{align}\label{eq:TL_delta}
\Delta_{n,j}=\lp X_{\sn(j)}-\frac{j}{n+1}\rp 
  f_x\lp   W_{n,j}\rp +\usd \lp X_{\sn(j)}-\frac{j}{n+1}\rp^2 f_{xx}\lp \delta_{n,j},W_{\sn(j)}\rp,
\end{align}
where  $\delta_{n,j}$ (resp. $\delta_{n,j+1}$) lies in the unordered segment $(X_{\sn(j)},j/(n+1))$ (resp. $(X_{\sn(j+1)},(j+1)/(n+1))$) and where $f_x$ and $f_{xx}$ are the first and second derivatives of $f$ with respect to the first coordinate.
This leads to expansions for $\Delta_{n,j}^2$ and  $\Delta_{n,j}\Delta_{n,j+1}$:
\begin{align*}
\Delta_{n,j}^2&
= \Bigl( X_{\sn(j)}-\frac{j}{n+1}\Bigr)^2  \Bigl( f_x\lp   W_{n,j}\rp
  +\usd \lp X_{\sn(j)}-\frac{j}{n+1}\rp f_{xx}\lp \delta_{n,j},W_{\sn(j)}\rp \Bigr)^2\\
\Delta_{n,j}&\Delta_{n,j+1}
= \lp X_{\sn(j)}-\frac{j}{n+1}\rp \lp X_{\sn(j+1)}-\frac{j+1}{n+1}\rp\\
&\times \Bigl( f_x\lp   W_{n,j}\rp
  +\usd \lp X_{\sn(j)}-\frac{j}{n+1}\rp f_{xx}\lp \delta_{n,j},W_{\sn(j)}\rp \Bigr)\\
  &\times \Bigl( f_x\lp   W_{n,j+1}\rp
  +\usd \lp X_{\sn(j+1)}-\frac{j+1}{n+1}\rp f_{xx}\lp \delta_{n,j+1},W_{\sn(j+1)}\rp\Bigr).
\end{align*}
Finally, using the boundedness of $f$, $f_x$, and $f_{xx}$, together with \eqref{eq:X_behav}, \eqref{eq:delta} follows.

Remark that the proof of \eqref{eq:delta} yields also
\begin{align}\label{eq:delta_seul}
\frac 1n \sum_{j=1}^{n-1} \Delta_{n,j}=O_{\P}\lp\frac {1}{\sqrt n}\rp,
\end{align}
from which it is clear that this term will contribute in the CLT on $Z_n$. Then \eqref{eq:delta} entails that the asymptotic study 
reduces to that of the empirical mean of 
$Z_{n,j}= B_{n,j}+ C_{n,j}$ 
where
\begin{align}
B_{n,j}\defeq \begin{pmatrix} 
f\lp   W_{n,j}\rp f\lp   W_{n,j+1}\rp\\ 
f(  W_{n,j})\\
f(  W_{n,j} )^2 
\end{pmatrix}  \text{ and } 
C_{n,j}\defeq \begin{pmatrix}
\Delta_{n,j} f\lp   W_{n,j+1}\rp+ \Delta_{n,j+1}f\lp   W_{n,j}\rp\\ 
\Delta_{n,j}\\
2\Delta_{n,j}f(  W_{n,j})
\end{pmatrix}\label{eq:decompBC}.
\end{align}

First, we consider $B_{n,j}$ in \eqref{eq:decompBC} and we establish the following result, the proof of which has been postponed to Appendix \ref{app:var_B}.

\begin{lem}\label{lem:B}
As $n\to \infty$, the random vector $  B_n$ given by
\[
\frac 1n \sum_{j=1}^{n-1} B_{n,j}=\frac 1n \sum_{j=1}^{n-1} \lp f\lp   W_{n,j}\rp f\lp  W_{n,j+1}\rp, f\lp   W_{n,j}\rp,f\lp   W_{n,j}\rp^2\rp^\top
\]
satisfies a CLT. More precisely, 
$
\sqrt n \bigl( B_n-m_B\bigr) \cvloi \mathcal{N}_3(0,\Sigma_B)$,
where 
\begin{align}\label{def:mB}
 m_B\defeq \lp \E[YY'],\E[Y],\E[Y^2]\rp^\top,
\end{align} 
$Y'=f(X,W')$, $W'$ is an independent copy of $W$,
and $\Sigma_B$ has an explicit expression given in Appendix \ref{app:var_B}.
\end{lem}

Remark that $Y'$ is the so-called Pick-Freeze version of $Y$ with respect to $X$. 
Secondly, we establish a conditional CLT for the empirical mean of the $C_{n,j}$'s  defined in \eqref{eq:decompBC}. The reader is referred to Appendix \ref{app:var_C} for the proof of this result.

\begin{lem}\label{lem:C}
There exists a measurable set $\Pi \in \Omega_W$ having $\P_W$-probability one such that, for any $\omega_W\in \Pi$, we have
\[
\sqrt n C_n(\cdot{},\omega_W) \overset{\mathcal{L}_X}{\underset{n\to\infty}{\longrightarrow}} \mathcal{N}_3(0,\Sigma_C).
\]
Moreover, $\Sigma_C$ does not depend on $\omega_W$ and has an explicit expression given Appendix \ref{app:var_C}.
\end{lem}

Considering the characteristic function of the vector $\sqrt n(B_n-\E[B_n],C_n)$, one may write
\begin{align*}
\E\left[e^{i(\sqrt n\pscal{s,(B_n-\E[B_n])}+\sqrt n\pscal{t,C_n})}\right]&=\E\left[e^{i\sqrt n\pscal{s,(B_n-\E[B_n])}}\E\left[e^{i\sqrt n\pscal{t,C_n}}\big\vert \mathcal F_W\right]\right]
\end{align*}
for any $s$ and $t\in \R^3$. 
On the one hand, $\E\left[e^{i\sqrt n\pscal{t,C_n} }\big\vert \mathcal F_W\right]$ converges a.s.\ to $\exp\{-t^\top \Sigma_C t /2\}$ which is not random. 
On the other hand, $\sqrt n\pscal{s,(B_n-\E[B_n])}$ converges in distribution to a Gaussian random variable denoted by $B_s$. By Slutsky's lemma, 
\[
\lp \sqrt n\pscal{s,(B_n-\E[B_n])}, \E\left[e^{i\sqrt n\pscal{t,C_n} }\big\vert \mathcal F_W\right] \rp
\]
converges in distribution to $(B_s,\exp\{-t^\top \Sigma_C t /2\})$. We consider the application $h\colon (u,v)\in \R\times D(0,1) \mapsto e^{iu}v \in \mathbb C$ where $D(0,1)$ is the unit disc in $\mathbb C$. The continuity and the boundedness of $h$ lead to the convergence in distribution of $e^{i\sqrt n\pscal{s,(B_n-\E[B_n])}}\left[e^{i\sqrt n\pscal{t,C_n}}\big\vert \mathcal F_W\right]$ and we conclude to the 
asymptotic normality of $\sqrt n(B_n-\E[B_n],C_n)$ to a six-dimensional Gaussian random vector with zero mean and variance-covariance matrix $\begin{pmatrix}
\Sigma_B & 0\\
0 & \Sigma_C\\
\end{pmatrix}$. It remains to apply the so-called delta method \cite[Theorem 3.1]{van2000asymptotic} and Slutsky's lemma to get the required result. The details of the computation of the asymptotic variance $\sigma^2$ can be found in  Appendix \ref{app:var_tcl}.

\subsection{Proof of Lemma \ref{lem:B}}\label{app:var_B}

One has 
\begin{align*}
&\E[B_n]=\frac 1n \sum_{j=1}^{n-1}  \lp
\E\left[f\lp  W_{n,j}\rp f\lp  W_{n,j+1}\rp\right],
 \E\left[f\lp  W_{n,j}\rp\right],
 \E\left[f\lp  W_{n,j}\rp^2\right]
\rp^\top,
\end{align*}
the first coordinate of which converges as $n\to \infty$ to 
\begin{align*}
\int \E\left[f\lp x,W\rp f\lp x',W'\rp\right]\text d \mathcal L_{(X,X)}(x,x')&=\int_0^1 \E\left[f\lp x,W\rp f\lp x,W'\rp\right]\text d x\\
&=\E\left[\E\left[f\lp X,W\rp f\lp X,W'\rp\vert X\right]\right]\\
&=\E\left[f\lp X,W\rp f\lp X,W'\rp\right]=\E\left[YY'\right].
\end{align*}
The two other coordinates can be handled similarly leading to 
\begin{align*}
&\E[B_n] \underset{n\to \infty}{\to}  
\lp \E[YY'],\E[Y],\E[Y^2]\rp^\top=m_B.
\end{align*}
We apply the CLT for dependent variables proved in \cite{orey1958central} to $\widetilde B_{n,j}^1$, the centered version of the random variables
$ f\bigl( W_{n,j}\bigr) f\bigl(  W_{n,j+1}\bigr)/{\sqrt n}$
 with $m=1$, $\alpha=0$, and 
 because $f$ is bounded (so is $\widetilde B_{n,j}^1$). 
Assumptions (1) and (2) in \cite{orey1958central} obviously hold, the assumption (3) is naturally fulfilled and assumption (4) is a mere consequence of Chebyshev's inequality and the boundedness of $f$. Now, it remains to check that assumption (5) holds. We have
\begin{align*}
\sum_{i,j=1}^{n-1} &\Cov( \widetilde B_{n,i}^1, \widetilde B_{n,j}^1)=\frac 1n \sum_{i,j=1}^{n-1} \Cov\left(f\lp  W_{n,i}\rp f\lp  W_{n,i+1}\rp,f\lp  W_{n,j}\rp f\lp  W_{n,j+1}\rp\right)\\
&=\frac 1n \sum_{j=1}^{n-1} \Var\left(f\lp  W_{n,j}\rp f\lp  W_{n,j+1}\rp\right) + \frac 2n \sum_{j=1}^{n-2} \Cov\left(f\lp  W_{n,j}\rp f\lp  W_{n,j+1}\rp,f\lp  W_{n,j+1}\rp f\lp  W_{n,j+2}\rp\right).
\end{align*}
On the one hand, by \cite[Lemma 1.1]{GGKL20-supp},
\begin{align*}
\frac 1n \sum_{j=1}^{n-1} & \Var\left(f\lp  W_{n,j}\rp f\lp  W_{n,j+1}\rp\right)
\underset{n\to \infty}{\to} \int \Var\left(f\lp x,W\rp f\lp x',W'\rp\right)\text d \mathcal L_{(X,X)}(x,x') \\
&= \int_0^1 \Var\left(f\lp x,W\rp f\lp x,W'\rp\right)\text dx = \E\left[\Var\left(f\lp X,W\rp f\lp X,W'\rp\vert X \right)\right]= \E\left[\Var\left(YY'\vert X\right)\right],
\end{align*}
where $W'$ is an independent copies of $W$, $Y=f(X,W)$, and $Y'=f(X,W')$.
On the other hand, by \cite[Lemma 1.1]{GGKL20-supp},
\begin{align*}
&\frac 1n \sum_{j=1}^{n-2} \Cov\left(f\lp  W_{n,j}\rp f\lp  W_{n,j+1}\rp,f\lp  W_{n,j+1}\rp f\lp  W_{n,j+2}\rp\right)\\
&\underset{n\to \infty}{\to}
\E\left[\Cov\left(f\lp X,W\rp f\lp X,W'\rp,f\lp X,W'\rp f\lp X,W''\rp\vert X\right)\right] =\E\left[\Cov\left(YY',YY''\vert X\right)\right],
\end{align*}
where $W'$ and $W''$ are two independent copies of $W$. Further, $Y=f(X,W)$, $Y'=f(X,W')$, and $Y''=f(X,W'')$.
Actually, notice that all linear combination of the coordinates of 
\begin{align}\label{def:one_dep}
\bigl( f(  W_{n,j}) f(  W_{n,j+1}), f(  W_{n,j}),f(  W_{n,j})^2\bigr)^\top
\end{align} 
is a one-dependent random variable. In addition, following the same lines as above, one may check that any linear combination still satisfies the assumptions of \cite{orey1958central}. Hence, any linear combination of the coordinates of $ B_n$ satisfies a CLT so that Lemma \ref{lem:B} is proved, up to the computation of the asymptotic variance-covariance matrix $\Sigma_B$ done in what follows.

\subsubsection*{Computation of the asymptotic covariance matrix \texorpdfstring{$\Sigma_B$}{f}}

We consider a linear combination of the random vector in \eqref{def:one_dep} given by
\begin{align*}
u f(  W_{n,j}) f(  W_{n,j+1})+v f(  W_{n,j})+wf(  W_{n,j})^2,
\end{align*} 
where $(u,v,w)\in \R^3$. This one-dimensional random vector is one-dependent and its centered version normalized by $\sqrt n$, denoted by ${\widetilde  B}_{n,j}$,  satisfies the assumptions of  \cite{orey1958central}. To calculate the asymptotic variance-covariance matrix $\Sigma_B$, we compute explicitly the limit of 
\[
\sum_{i,j=1}^{n-1} \Cov({\widetilde B}_{n,i},{\widetilde B}_{n,j}),
\]
as $n\to \infty$ using \cite[Lemma 1.1]{GGKL20-supp}. 
It remains to take $(1,0,0)$, $(0,1,0)$ and $(0,0,1)$ to get the diagonal terms of the asymptotic variance-covariance matrix and to solve a three-dimensional system of equations to get the remaining terms. Finally, as computed previously and using notation of \cite[Lemma 1.1]{GGKL20-supp}, the first diagonal term of $\Sigma_B$ is : 
\begin{align*}
&\Sigma_B^{1,1}=  \int \Var\left(f\lp x,W\rp f\lp x',W'\rp\right)\text d \mathcal L_{(X,X)}(x,x')\\
& \quad + 2\int \Cov\left(f\lp x,W\rp f\lp x',W'\rp,f\lp x',W'\rp f\lp x'',W''\rp\right)\text d \mathcal L_{(X,X,X)}(x,x',x'')\\
&=  \int_0^1 \Var\left(f\lp x,W\rp f\lp x,W'\rp\right)\text dx + 2\int_0^1 \Cov\left(f\lp x,W\rp f\lp x,W'\rp,f\lp x,W'\rp f\lp x,W''\rp\right)\text dx\\
&=\E\left[\Var\left(f\lp X,W\rp f\lp X,W'\rp|X\right) \right] + 2\E\left[\Cov\left(f\lp X,W\rp f\lp X,W'\rp,f\lp X,W'\rp f\lp X,W''\rp|X\right)\right]\\
&=\E\left[\Var\left(YY'|X\right) \right]
+ 2\E\left[\Cov\left(YY',YY''|X\right)\right],
\end{align*}
where we remind that $Y=f(X,W)$, $Y'=f(X,W')$, and $ Y''=f(X,W'')$ with $W'$ and $W''$ independent copies of $W$. 
The other terms are
%
\begin{align*}
\Sigma_B^{2,2}&=  
\int_0^1 \Var\left(f\lp x,W\rp \right)\text d x
=\E\left[\Var\left(f\lp X,W\rp \vert X\right)\right]
=\E\left[\Var(Y \vert X)\right],\\
\Sigma_B^{3,3}&=  \int_0^1 \Var\left(f\lp x,W\rp^2 \right)\text d x=\E\left[\Var\left(Y^2\vert X\right)\right],\\
\Sigma_B^{1,2}&=\Sigma_B^{2,1}=2\int_0^1 \Cov\left(f\lp x,W\rp f\lp x,W'\rp,f\lp x,W\rp \right)\text dx= 2 \E\left[\Cov\left(YY',Y\vert X\right)\right],\\
\Sigma_B^{1,3}&= \Sigma_B^{3,1}= 2 \int_0^1 \Cov\left(f\lp x,W\rp f\lp x,W'\rp,f\lp x,W\rp^2 \right)\text dx
 =2 \E\left[\Cov\left(YY',Y^2\vert X\right)\right],\\
\Sigma_B^{2,3}&=\Sigma_B^{3,2} =\int_0^1 \Cov\left(f\lp x,W\rp,f\lp x,W\rp ^2\right)\text d x=\E\left[\Cov(Y,Y^2\vert X)\right].
\end{align*}

\subsection{Proof of Lemma \ref{lem:C}}\label{app:var_C}

Let $\omega_W \in\Pi$ as defined in \cite[Lemma 1.1]{GGKL20-supp}. The aim is to establish a CLT for $\sqrt n C_{n,j}(\cdot{},\omega_W)$. To ease the reading, we omit the notation $(\cdot{},\omega_W)$ as classically done in probability.
First, dealing with the first coordinate $ f\lp   W_{n,j+1}\rp \Delta_{n,j} +f\lp   W_{n,j} \rp \Delta_{n,j+1} $ of $C_{n,j}$ defined in \eqref{eq:decompBC}, one has 
\begin{align*}
 f\lp   W_{n,j+1}\rp \Delta_{n,j} 
 =  &\lp X_{\sn(j)}-\frac{j}{n+1}\rp   f\lp   W_{n,j+1}\rp
  f_x\lp   W_{n,j}\rp \\
& + \usd   \lp X_{\sn(j)}-\frac{j}{n+1}\rp^2 f\lp   W_{n,j+1}\rp f_{xx}\lp \delta_{n,j},W_{j}\rp 
\end{align*}
using the expansion of  $\Delta_{n,j}$ given in \eqref{eq:TL_delta}.
By \eqref{eq:X_behav} and using the boundedness of $f$ and $f_{xx}$, we get that
\begin{align*}
\frac 1n \sum_{j=1}^{n-1} &\lp X_{\sn(j)}-\frac{j}{n+1}\rp^2  f\lp   W_{n,j+1}\rp    f_{xx}\lp \delta_{n,j},W_{j}\rp
\end{align*}
is $O_{\P}\lp 1/n\rp$. 
We follow the same lines to treat the term $f\lp   W_{n,j} \rp \Delta_{n,j+1}$
and thus 
\begin{align*}
\frac 1n & \sum_{j=1}^{n-1} 
  f\lp   W_{n,j+1}\rp \Delta_{n,j} +f\lp   W_{n,j} \rp \Delta_{n,j+1}= \frac 1n \sum_{j=1}^{n-1}  \lp X_{\sn(j)}-\frac{j}{n+1}\rp   f\lp   W_{n,j+1}\rp
  f_x\lp   W_{n,j}\rp \\
  & \quad + \frac 1n \sum_{j=1}^{n-1}  \lp X_{\sn(j+1)}-\frac{j+1}{n+1}\rp   f\lp   W_{n,j}\rp
  f_x\lp   W_{n,j+1}\rp +  O_{\P}\lp \frac 1n\rp\\
 &= 
\frac 1n \sum_{j=1}^{n-1} 
\lp X_{\sn(j)}-\frac{j}{n+1}\rp  f_x\lp   W_{n,j}\rp
\lp f\lp   { W_{n,j-1}} \rp+ f\lp   W_{n,j+1}\rp \rp + O_{\P}\lp \frac 1n\rp.
\end{align*}
So that, using again the expansion of  $\Delta_{n,j}$ given in \eqref{eq:TL_delta}, \eqref{eq:X_behav}, and the boundedness of $f$ and $f_{xx}$ to handle the second and third coordinate of $C_{n,j}$, the study of $C_n$ reduces to that of the random vector
\begin{align}\label{eq:C}
\frac 1n &\sum_{j=1}^{n-1} 
\lp X_{\sn(j)}-\frac{j}{n+1}\rp  f_x\lp   W_{n,j}\rp \begin{pmatrix}
f\lp   { W_{n,j-1}} \rp+ f\lp   W_{n,j+1}\rp
\\
1\\
2f\lp   W_{n,j+1}\rp
\end{pmatrix}
\end{align}
by the independence between $\sigma_n$ and $W_1,\ldots, W_n$.
In that view, let us consider the following linear combination 
$u (f( { W_{n,j-1}})+f(  W_{n,j+1}))+v +2wf(  W_{n,j+1})$,
where $(u,v,w)\in \R^3$ and the empirical mean 
\begin{align}
\frac 1n \sum_{j=1}^{n-1} 
\big( X_{\sn(j)}-&\frac{j}{n+1}\big)  f_x\lp   W_{n,j}\rp \times \lp u (f(  {W_{n,j-1}}) +f(  W_{n,j+1}))+v +2wf(  W_{n,j+1})\rp.\label{eq:emp_mean}
\end{align} 
Now it remains to apply \cite[Lemma 1.4]{GGKL20-supp} \footnote{A slightly generalization of this lemma is required to handle the pair $(j/(n+1),(j+1)/(n+1))$ rather than the quantity $j/n$. Its proof comes directly following the same lines as in the proof of this lemma} with {$\chi_{j}=\lp  W_{j-1}, W_j,W_{j+1}\rp$}
and $\psi=\psi_{uvw}$ with 
\begin{align}\label{def:psi_lemC}
\psi_{uvw}\lp \frac{j-1}{n+1},\frac{j}{n+1},\frac{j+1}{n+1},\chi_j\rp=f_x\lp   W_{n,j}\rp \lp u (f(  {W_{n,j-1}} ) +f(  W_{n,j+1}))+v +2wf(  W_{n,j+1})\rp,
\end{align}
noticing that, as $n\to \infty$, 
$(1/n)\sum_{j=1}^{n-1}\delta_{(j-1)/(n+1),j/(n+1),(j+1)/(n+1),\chi _{j}}$  converges in distribution to $Q= \mathcal L_{(X,X,X)} \otimes \mathcal L_W \otimes \mathcal L_W  \otimes \mathcal L_W$
by \cite[Lemma 1.1]{GGKL20-supp}. Thus we deduce that the empirical mean 
in \eqref{eq:emp_mean} converges in distribution for any 3-uplet $(u,v,w)$. 
Since any linear combination of the components of the random vector defined in \eqref{eq:C} satisfies a CLT, so does the random vector itself. The proof of Lemma \ref{lem:C} is now complete, up to the computation of the asymptotic variance-covariance matrix $\Sigma_C$ done in the paragraph that follows.

\subsubsection*{Computation of the asymptotic covariance matrix \texorpdfstring{$\Sigma_C$}{f}}

We use the explicit expression (4) in the proof of \cite[Lemma 1.4]{GGKL20-supp} of the asymptotic variance $\sigma_{\psi}^2$ (actually a slightly generalized version of the lemma) with {$Q = \mathcal L_{(X,X,X)} \otimes \mathcal L_W \otimes \mathcal L_W \otimes \mathcal L_W$} and 
with $\psi$ given by \eqref{def:psi_lemC}. Then taking the values
$(1,0,0)$, $(0,1,0)$ and $(0,0,1)$ leads to the diagonal terms of the asymptotic variance-covariance matrix $\Sigma_C$ while solving a three-dimensional system of equations provides the remaining terms. For instance,
reminding that {$\chi_j=(W_{j-1},W_j,W_{j+1})$} and $W_{n,j}=(j/(n+1),W_j)$ and 
\[
\psi_{100}\left(\frac{j-1}{n+1},\frac{j}{n+1},\frac{j+1}{n+1},\chi_j \right)=f_x\lp  W_{n,j} \rp  (f( {W_{n,j-1}}) +f(  W_{n,j+1}))
\]
(namely, $\psi_{uvw}$ with $(u,v,w)=(1,0,0)$), we have
\begin{align*}
\Sigma_C^{1,1}=&\int \psi_{100}(x_1,x'_1,x''_1,\chi_1)\psi_{100}(x_2,x'_2,x''_2,\chi_2) x_1\wedge x_2\wedge x_1'\wedge x_2'\wedge x_1''\wedge x_2'' \nonumber\\
&\times \text d Q(x_1,x'_1,x''_1,\chi_1) \text d Q(x_2,x'_2,x''_2,\chi_2) -\left(\int \psi_{100}(x,x',x'',\chi) x\wedge x'\wedge x'' \text d Q(x,x',x'',\chi)\right)^2\nonumber\\
=&\E[(Y_1+Y_1')(Y_2+Y_2')f_x(X_1,W_1) f_x(X_2,W_2) (X_1\wedge X_2)]-\E[(Y+Y')f_x(X,W) X]^2,
\end{align*}
where we remind that $Y=f(X,W)$ and $Y'=f(X,W')$ with $W'$ an independent copy of $W$ (and analogously for $Y_1$ and $Y_2$). 
Finally, the remaining terms of $\Sigma_C$ are:
\begin{align*}
\Sigma_C^{2,2}&=\E[f_x(X_1,W_1) f_x(X_2,W_2) (X_1\wedge X_2)]-\E[f_x(X,W) X]^2\\
\Sigma_C^{3,3}&=4\E[Y_1'Y_2'f_x(X_1,W_1) f_x(X_2,W_2) (X_1\wedge X_2)]-4\E[Y'f_x(X,W) X]^2\\
\Sigma_C^{1,2}&=\Sigma_C^{2,1}=\E[(Y_1+Y'_1)f_x(X_1,W_1)f_x(X_2,W_2) (X_1\wedge X_2)]-\E[(Y+Y')f_x(X,W) X]\E[f_x(X,W) X]\\
\Sigma_C^{1,3}&=\Sigma_C^{3,1}=2\E[(Y_1+Y'_1)f_x(X_1,W_1) Y_2'f_x(X_2,W_2) (X_1\wedge X_2)]-2\E[(Y+Y')f_x(X,W) X]\E[Y'f_x(X,W) X]\\
\Sigma_C^{2,3}&=\Sigma_C^{3,2}=2\E[f_x(X_1,W_1)Y_2' f_x(X_2,W_2)  (X_1\wedge X_2)]-2\E[f_x(X,W) X]\E[Y'f_x(X,W) X].
\end{align*}

\subsection{Asymptotic variance \texorpdfstring{$\sigma^2$}{f} of Theorem \ref{th:tcl}}\label{app:var_tcl}

We have proved yet that
\[ 
\sqrt n \left( \begin{pmatrix}
B_n\\
C_n\\
\end{pmatrix}
 - 
\begin{pmatrix}
m_B\\
0\\
\end{pmatrix}  
 \right)\cvloi
		 \mathcal N_6 \left( 0, \begin{pmatrix}
\Sigma_B & 0\\
0 & \Sigma_C\\
\end{pmatrix}
\right),
\]
where the explicit expressions of $m_B$, $\Sigma_B$ and $\Sigma_C$ are given in \eqref{def:mB} of Lemma \ref{lem:B}, Appendices \ref{app:var_B} and \ref{app:var_C} respectively.
Applying the so-called delta method \cite[Theorem 3.1]{van2000asymptotic} to the linear function $f(x,y)=x+y$, we conclude that
\begin{align}\label{eq:CLT_somme}
\sqrt n ( Z_n - m_B)\cvloi
		 \mathcal N_3 \left( 0, 
\Sigma_B + \Sigma_C
\right)
\end{align}
Further, we notice that
$ 
\xi_n^{\text{Sobol'}}(X,Y) \egloi \Psi(Z_n) 
$ 
with $\Psi (x,y,z) = ( x-y^2 )/(z - y^2)$. 	 
The so-called delta method \cite[Theorem 3.1]{van2000asymptotic} then gives
\[ 
\sqrt N \left( \xi_n^{\text{Sobol'}}(X,Y) - S^X\right) \cvloi \mathcal N_1(0, \sigma^2) \]
where $S^X= \Var(\E[Y|X])/\Var(Y)$ is the first-order  Sobol' index with respect to $X$ and $\sigma^2=g^\top (\Sigma_B+\Sigma_C) g$ with
$
g =  \nabla \Psi(m_B)$.
By assumption $\Var(Y) \neq 0$, $\Psi$ is differentiable at $m_B$ and we will see in the sequel that $g^\top (\Sigma_B+\Sigma_C) g \neq 0$, so that the application of the delta method is justified. By differentiation, we get that, for any $x$, $y$, and $z$ so that $z \neq y^2$:
\begin{align}\label{eq:nabla_psi}
\nabla \Psi ( x, y, z) =\left( \frac{1}{z-y^2}, -2y \frac{z-x}{ (z-y^2)^2 },- \frac{x-y^2}{(z-y^2)^2} \right)^\top 
\end{align}
so that
\begin{align*}
g = \nabla \Psi(m_B)&=\left( \frac{1}{\Var(Y)}, 
 2\E[Y]\frac{\E[YY']-\E[Y^2]}{\Var(Y)^2},
 - \frac{S^X}{\Var(Y)} \right)^\top=\frac{1}{\Var(Y)} \left( 1, 
 2\E[Y](S^X-1),
 - S^X \right)^\top.
\end{align*} 
Hence the asymptotic variance $\sigma^2$ in Theorem \ref{th:tcl} is finally given by $
	\sigma^2 = g^\top \lp \Sigma_B + \Sigma_C\rp  g $
where $\Sigma_B$ and $\Sigma_C$ have been defined in Appendices \ref{app:var_B} and \ref{app:var_C} respectively. The matrix $\Sigma_B$ rewrites as
\[
\Sigma_B=\begin{pmatrix}
v_{01}+2c_{01,02} & 2c_{01,03} & 2c_{01,00}\\
2c_{01,03} & \Var(Y)(1-S^X) & 2c_{03,00}\\
2c_{01,00} & 2c_{03,00} & v_{00} 
\end{pmatrix}
\]
where
$v_{ij}=\E[\Var(A_iA_j|X)]$, $
c_{ij,kl}=\E[\Cov(A_iA_j,A_kA_l|X)]$, 
 $A_0=Y$, $A_1=Y'$, $A_2=Y''$, and $A_3=1$ ($Y$ and $Y''$ have been defined just before \eqref{def:one_dep}).
The matrix $\Sigma_C$ rewrites as
\[
\Sigma_C=\begin{pmatrix}
s^2_{\psi_{100}} & s^2_{\psi_{110}} & s^2_{\psi_{101}}\\
s^2_{\psi_{110}} & s^2_{\psi_{010}} & s^2_{\psi_{011}}\\
s^2_{\psi_{101}} & s^2_{\psi_{011}} & s^2_{\psi_{001}} 
\end{pmatrix}
\]
where $s^2_{\psi}$ and $\psi_{uvw}$ have been defined in \cite[Equation (4)]{GGKL20-supp} and \eqref{def:psi_lemC} respectively. 

\section{Proof of the asymtotic efficiency of \texorpdfstring{$R_n^1$}{f}}\label{app:ae}

\begin{proof}[Proof of Proposition \ref{prop:ae}]
By \cite[Theorems 3.4 and 3.5]{da2008efficient} and classical results on efficiency, observe that
\begin{align*}
 U_n = \left( \widehat T_n ,\frac{1}{n} \sum_{i=1}^n Y_i, \frac{1}{n} \sum_{i=1}^n Y_i^2 \right)^\top
 \end{align*}
is asymptotically efficient, componentwise, for estimating 
$U = \left( \E[\E[Y\vert X]^2], \E[Y],\E[Y^2] \right)^\top$.
The efficiency in product space \cite[Theorem 25.50]{van2000asymptotic} yields the joint efficiency from this componentwise efficiency. Now, we consider once again the function $\Psi$ introduced in the proof of Theorem \ref{th:tcl}. Since $\Psi$ is differentiable on
$\R^3 \setminus \left\{(x,y,z)\, \big|\, z \neq y^2 \right\}$, the efficiency and delta method result \cite[Theorem 25.47]{van2000asymptotic} implies that $\left( \Psi\left( U_n\right) \right )_n$ is asymptotically efficient for estimating $\Psi(U)$. The conclusion follows as $\Psi(U)=S^X$.

Let us compute the minimal variance. To do so, assume that the joint distribution $P$ of $(X,Y)$ is absolutely continuous with respect to the Cartesian product $P_X\otimes P_Y$, namely $P(dx,dy)=f(x,y)P_X(dx) P_Y(dy)$. Then 
\begin{align*}
\E[Y\vert X=x]&=\int y f_{Y\vert X=x}(y)P_Y(dy)= \int y \frac{f(x,y)}{\int f(x,y)P_Y(dy)}P_Y(dy).
\end{align*}
For any $t\in (0,1)$, let us introduce $f_t(x,y)\defeq (1+th(x,y)) f(x,y)$  and
\[
P_t(dx,dy)\defeq (1+th(x,y)) f(x,y)P_X(dx) P_Y(dy)
\]
where $h(x,y)>-1$ and 
$
\int h(x,y) f(x,y)P_x(dx) P_Y(dy)=0$.
Now we consider the function
\begin{align*}
F(t)&\defeq \iint_{x,y'} \left(\frac{\int y f_t(x,y)P_Y(dy)}{\int  f_t(x,y)P_Y(dy)}\right)^2 P_t(dx,dy').
\end{align*}
Denoting by $G(x,t)\defeq \int y f_t(x,y)P_Y(dy)/\int  f_t(x,y)P_Y(dy)$, one gets 
\begin{align*}
F'(t)&=\iint_{x,y'} \left[2 G(x,t) \frac{\partial }{\partial t} G(x,t) f_t(x,y') +G(x,t)^2 h(x,y')f(x,y') \right] P_X(dx) P_Y(dy')
\end{align*}
so that $
F'(0)
=\pscal{\E[Y\vert X=x](2y-\E[Y\vert X=x]),h}_P$.
The interest function $
I\defeq \E[Y\vert X](2Y-\E[Y\vert X])
$
has $\E[\E[Y\vert X]^2]$ and variance
$
\Var(\E[Y\vert X](2Y-\E[Y\vert X]))$.
Hence it remains to apply the delta method to get the final (minimal) variance
\[
g^\top \begin{pmatrix}
\Var(I) & \Cov(I,Y) & \Cov(I,Y^2)\\
\Cov(I,Y) & \Var(Y) & \Cov(Y,Y^2)\\
\Cov(I,Y^2) & \Cov(Y,Y^2) & \Var(Y^2)
\end{pmatrix}  g 
\]
where $g \defeq  \nabla \Psi(U)$, and by \eqref{eq:nabla_psi},
\begin{align*}
g = \left( \frac{1}{\Var(Y)}, 
 2\E[Y]\frac{\E[\E[Y\vert X]^2]-\E[Y^2]}{\Var(Y)^2},
 - \frac{S^X}{\Var(Y)} \right)^\top=\frac{1}{\Var(Y)} \left( 1, 
 2\E[Y](S^X-1),
 - S^X \right)^\top.
\end{align*} 
Finally, one gets the minimal variance mentioned in Proposition \ref{prop:ae}.
\end{proof}

\begin{rem}
This result can be also obtained making a LAN perturbation of the functional derivative on the tangent space. In this setting and following the notation of \cite[Chapitre 25]{van2000asymptotic}, 
let us consider the functional $\Phi$ defined by
\[
\Phi(P)\defeq \frac{\E_P[\E_P[Y|X]]-\E_P[Y]^2}{\E_P[Y^2]-\E_P[Y]^2}. 
\]
Then, with the notation $P_t$ for $t\in (0,1)$ introduced in the above proof, one gets
\begin{align*}
\frac{d}{dt} \Phi(P_t)\restrict{t=0}
&=\frac{1}{\Var(Y)}\pscal{\E[Y|X](2Y-\E[Y|X])-2\E[Y] Y-S^X(Y^2-2\E[Y]Y),h}_P
\end{align*}
leading to 
$\tilde \Phi\defeq \frac{1}{\Var(Y)}\left(2\E[Y]  Y (1-S^X)+S^X Y^2 -\E[Y|X](\E[Y|X]-2Y)\right)$
and the minimal variance is given by
$
\sigma_{\min}^2
= \Var(\tilde \Phi)=\frac{1}{\Var(Y)^2}\Var\left(2\E[Y](1-S^X)Y+S^XY^2+\E[Y\vert X](\E[Y\vert X]-2Y)\right)$
that coincides with the expression obtained via the delta method  in Proposition \ref{prop:ae}.
\end{rem}

%
%
%

\bibliographystyle{abbrv}
\bibliography{biblio_SA}

\includepdf[pages=-]{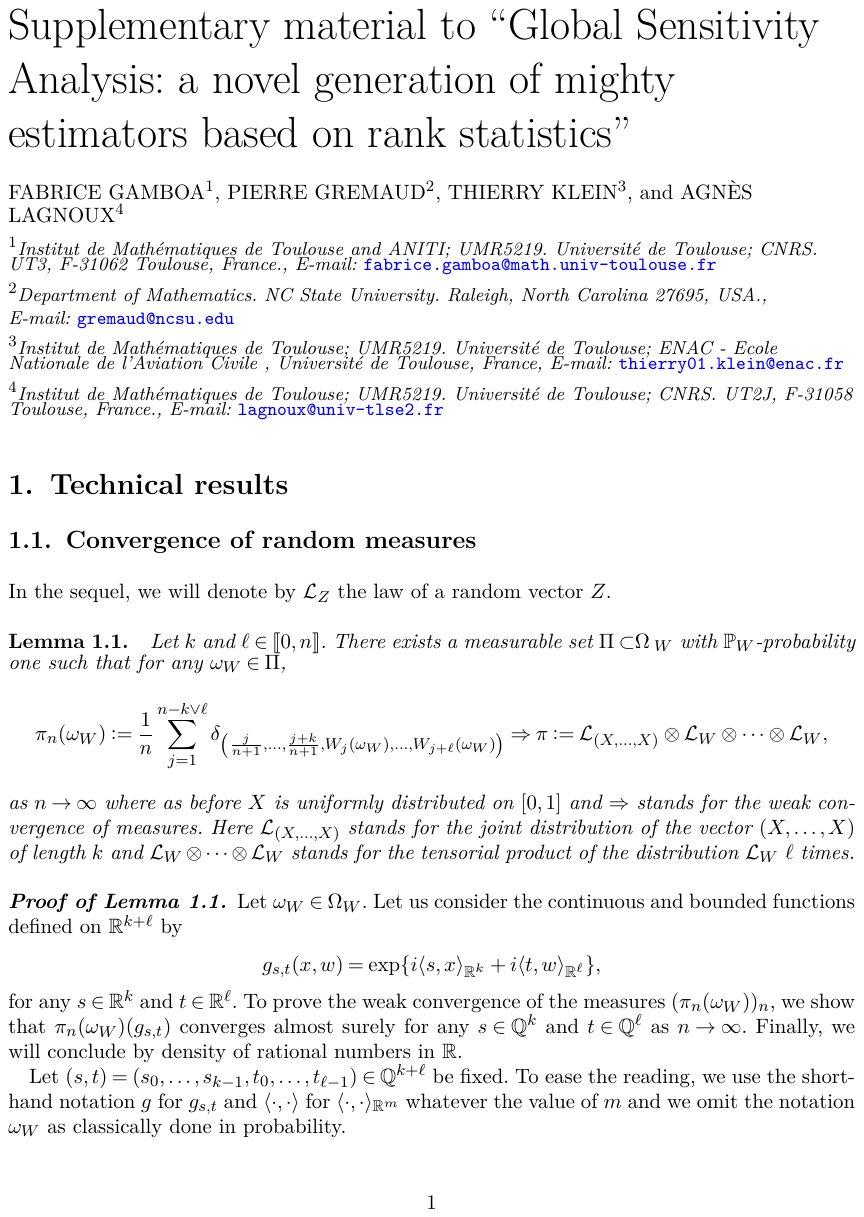}

\includepdf[pages=-]{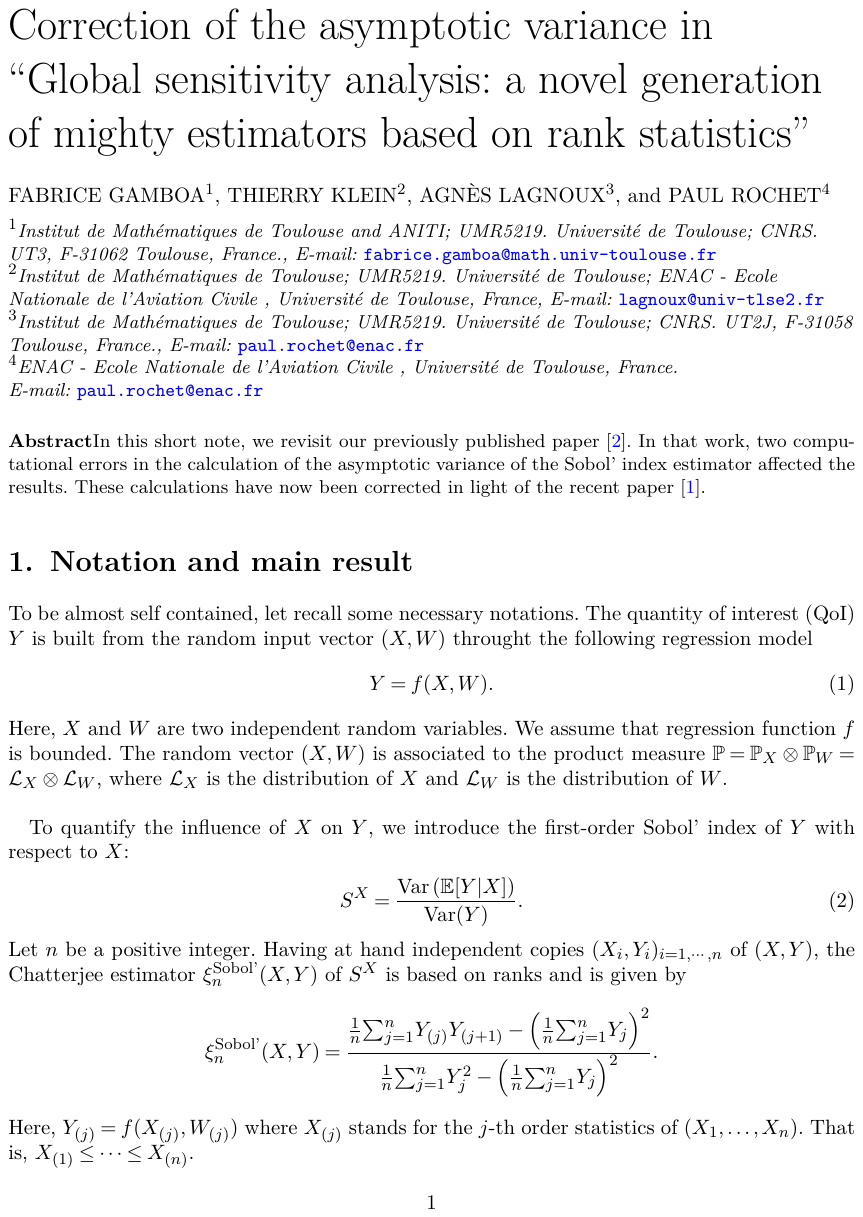}

\end{document}